\documentclass{birkjour}
 \newtheorem{thm}{Theorem}[section]
 
 \newtheorem{lem}[thm]{Lemma}
 \newtheorem{prop}[thm]{Proposition}
 \theoremstyle{definition}
 \newtheorem{defn}[thm]{Definition}
 \theoremstyle{remark}
 \newtheorem{rem}[thm]{Remark}
 \newtheorem*{ex}{Example}
 \numberwithin{equation}{section}

\begin{document}

%
%
%
%
%
%
%
%
%

\title[The Referential Gradient]
 {On the Geometric Structure of Flows I: \\ The Referential Gradient. A Generally Covariant Measure of Flow Geometry}

\author[J.K. Edmondson]{Justin K. Edmondson}

\address{%
Department of Aerospace Engineering \\
University of Michigan\\
1320 Beal Avenue \\
Ann Arbor, MI 48105}

\email{jkedmond@umich.edu}

\thanks{This research was supported, in part, by NSF Grant AGS-1043012, and NASA LWS Grant NNX10AQ61G.}
\subjclass{Primary 70G45, 37C10; Secondary 53C44, 53C15}

\keywords{Differential-geometric methods, Vector fields, flows, ordinary differential equations; Geometric evolution equations, General geometric structures on manifolds}

\date{\today}

\begin{abstract}
Assuming \textit{a-priori} a smooth generating vector field, we introduce a generally covariant measure of the flow geometry called the referential gradient of the flow. The main result is the explicit relation between the referential gradient and the generating vector field, and is provided for from two equivalent perspectives: a Lagrangian specification with respect to a generalized parameter, and an Eulerian specification making explicit the evolution dynamics. Furthermore, we provide explicit non-trivial conditions which govern the transformation properties of the referential gradient object.
\end{abstract}

\maketitle

\section{Introduction}
\label{Introduction}

The present work concerns the geometric structure of the flow of a vector field in four-dimensional spacetime. We work from the perspective that  the generating vector field satisfying some set of governing evolution equations is the primary quantity, and the flow is the secondary (or derived) quantity. Assuming \textit{a-priori} a smooth generating vector field, we introduce a generally covariant measure of the flow geometry called the \textit{Referential Gradient of the Flow}. The main result of this work is the explicit relation between the referential gradient of the flow and the generating vector field, and is provided for from two equivalent perspectives: a Lagrangian specification with respect to a generalized parameter, and an Eulerian specification making explicit the evolution dynamics. Furthermore, we provide explicit non-trivial conditions which govern the transformation properties of the referential gradient object.

The layout of this paper is as follows. Section (\ref{S:FlowAndRefGrad}) provides the standard differential geometry context (\ref{SS:Preliminaries}) and formalism of mathematical flow representations (\ref{SS:GeneralConnectivityMap}) in order to provide a rigorous framework from which to define the referential gradient. The main results of the paper are put forth in section (\ref{S:RefGradDefined}) where we define the referential gradient of the flow. In section (\ref{SS:LagrangeSpecRefGrad}), we prove the Lagrangian specification Theorem (\ref{Thm:RefGradLagrangian}) from which is given a non-local, closed-form functional solution with respect the generating vector field. Furthermore, due to the non-local nature of the Lagrangian specification of the referential gradient, in section (\ref{SS:RefGradTransforms}) we prove three lemmas that identify the non-trivial referential gradient transformation conditions: Lemma (\ref{lem:RefGrad}) provides the condition for manifest covariance; Lemma (\ref{lem:RefGradGroup}) provides for the group property with respect to the connectivity parameter; and Lemma (\ref{lem:RepresentationRelations}) provides the proper relations between the corresponding referential gradient representations associated with a change of integration variable. In section (\ref{SS:RefGradDynamics}), we prove the Eulerian specification Theorem (\ref{Thm:RefGradDynamics}) which in a coordinate chart  makes explicit the referential gradient dynamics at each point of the manifold. 

Throughout the paper, we work in both coordinate-free language, show coordinate-dependent expressions, and give explicit illustrative examples. Greek indices $\lbrace \ 0 \dots 3 \ \rbrace$ denote spacetime components, and Latin indices $\lbrace \ 1 \dots 3 \ \rbrace$ denote only spatial components; in addition we employ the standard Einstein summation convention.

\section{The Flow of a Vector Field}
\label{S:FlowAndRefGrad}

\subsection{Preliminaries}
\label{SS:Preliminaries}

Let $\left( \ \mathbb{M} , g \ \right)$ be a simply connected region of spacetime with (semi-Riemannian) metric $g$. Furthermore, assume the spacetime topology may be foliated as $\mathbb{M} \simeq \Sigma \times \mathbb{R}$, where $\Sigma$ is a three-dimensional spatial hypersurface, and $\mathbb{R}$ is the time axis. Let $\left( \ \mathbb{U}_{I} , x_{I} \ \right)_{I \in \mathcal{I}}$ be an atlas of $\mathbb{M}$, such that an arbitrary point $p \in \mathbb{U}_{I}$ has coordinate component functions $x_{I} \left( p \right) = x^{\mu}_{I} \left( p \right) \in \Sigma \times \mathbb{R}$, where $x^{0}_{I}$ and $x^{i}_{I}$ are identified with coordinate time and 3-space, respectively.

Given the atlas $\left( \ \mathbb{U}_{I} , x_{I} \ \right)_{I \in \mathcal{I}}$ over $\mathbb{M}$, we have on each chart $\mathbb{U}_{I}$ the local coordinate basis $\hat{e}^{I}_{\mu} \equiv  \partial^{I}_{\mu}$ and local dual basis $\hat{e}^{\mu}_{I} \equiv dx^{\mu}_{I}$, defined, respectively, via $\partial^{I}_{\mu} \bigl( \ x^{\nu}_{I} \left( p \right) \ \bigr) = \delta^{\nu}_{\mu}$ and $dx^{\mu}_{I} \Bigl( \ \partial^{I}_{\nu} \bigl( \ x_{I} \left( p \right) \ \bigr) \ \Bigr) = \delta^{\mu}_{\nu}$, where $\delta^{\mu}_{\nu}$ is the Kronecker delta. The invariant interval is given in coordinates by,
\begin{equation} \label{E:Metric}
\displaystyle
ds^{2} = g_{\mu \nu} \ dx^{\mu}_{I} \ dx^{\nu}_{I} \\
\end{equation}

\noindent Where in general, in the coordinate chart $\mathbb{U}_{I}$ the metric components are given functions of position, $g_{\mu \nu} = g_{\mu \nu} \bigl( \ x_{I} \left( p \right) \ \bigr)$.

We assume an affine connection $\nabla : T^{m}_{n} \left( \mathbb{M} \right) \rightarrow T^{m}_{n+1} \left( \mathbb{M} \right)$, that satisfies the standard axioms (see e.g., Ref. \cite{Wald1984} $\S$ 3.1). In a coordinate chart $\mathbb{U}_{I}$, the connection components are defined via the action on basis vectors, $\nabla_{\nu} \ \hat{e}^{I}_{\mu} = \Gamma^{\rho}_{\mu \nu} \ \hat{e}^{I}_{\rho}$ and $\nabla_{\nu} \ \hat{e}^{\mu}_{I} = - \Gamma^{\nu}_{\mu \rho} \ \hat{e}^{\rho}_{I}$, such that the covariant derivative of a mixed tensor field $t \left( p \right) \in T^{m}_{n} \left( \mathbb{M} \right)$ is given by,
\begin{equation} \label{E:TensorCovDeriv}
\displaystyle
\nabla_{\nu} \ t^{\eta_{1} \dots \eta_{m}}_{\beta_{1} \dots \beta_{n}} = \frac{\partial}{\partial x^{I}_{\nu}} \ t^{\eta_{1} \dots \eta_{m}}_{\beta_{1} \dots \beta_{n}} + \sum_{k = 1}^{m} \ \Gamma^{\eta_{k}}_{\nu \rho} \ t^{\eta_{1} \dots \rho \dots \eta_{m}}_{\beta_{1} \dots \beta_{n}}  - \sum_{l = 1}^{n} \ \Gamma^{\rho}_{\nu \beta_{l}} \ t^{\eta_{1} \dots \eta_{m}}_{\beta_{1} \dots \rho \dots \beta_{n}} \\
\end{equation}

\noindent In addition, the covariant derivative with respect to $\nabla$ in the direction of the vector field $w \in T^{1}_{0} \left( \mathbb{M} \right)$ is given by, $\nabla_{w} = i_{w} \circ \nabla : T^{m}_{n} \left( \mathbb{M} \right) \rightarrow T^{m}_{n} \left( \mathbb{M} \right)$, where $i_{w}$ is the contraction map (with some abuse of notation where the meaning is clear, we write the contraction map $i_{w} \circ \nabla = w \cdot \nabla$, interpreted as the standard scalar product).

Beyond what is physically reasonable, we make no assumptions regarding the metric or the connection. Without loss of generality, one may perform particular calculations using the Levi-Civita connection with components given by,
\begin{equation} \label{E:ChristoffelSymbols}
\displaystyle
\Gamma^{\rho}_{\mu \nu} = \frac{1}{2} \ g^{\rho \sigma} \ \biggl( \ \frac{\partial g_{\sigma \mu}}{\partial x^{\nu}_{I}} + \frac{\partial g_{\sigma \nu}}{\partial x^{\mu}_{I}} - \frac{\partial g_{\mu \nu}}{\partial x^{\sigma}_{I}} \ \biggr) \\
\\
\end{equation}

\subsection{The Flow of a Vector Field}
\label{SS:GeneralConnectivityMap}

Let $B: \mathbb{M} \rightarrow T \mathbb{M}$ be a smooth vector field everywhere on $\mathbb{M}$ (unless explicitly stated otherwise, by $T \mathbb{M}$ we mean $T^{1}_{0} \left( \mathbb{M} \right)$). In a chart $\mathbb{U}_{I}$, the vector field may be written,
\begin{equation} \label{E:GenVectField}
\displaystyle
B \left( p \right) = B^{\mu}_{I} \bigl( \ x_{I} \left( p \right) \ \bigr) \ \hat{e}^{I}_{\mu} \\
\end{equation}

\noindent Where the component functions $B^{\mu}_{I} \bigl( \ x_{I} \left( p \right) \ \bigr) \in C^{\infty} \left( \mathbb{M} \right)$. 

\begin{ex}
Identify the vector field $B \left( p \right)$ with a four-velocity field in Minkowski spacetime, curvilinear coordinates (see e.g., Ref. \cite{Tsamparlis2010} $\S$ 6.2), the contravariant components of which are given by,
\begin{equation} \label{Ex:FourVelocity}
\displaystyle
U^{\mu}_{I} \bigl( \ x_{I} \left( p \right) \ \bigr) = \Bigl( \ U^{0}_{I} \bigl( \ x_{I} \left( p \right) \ \bigr) \ , \ U^{i}_{I} \bigl( \ x_{I} \left( p \right) \ \bigr) \ \Bigr) \\
\end{equation}
\noindent Where $U^{0}_{I} \bigl( \ x_{I} \left( p \right) \ \bigr)$ is the velocity field in a co-moving coordinate frame.
\end{ex}

\begin{ex}
Identify the vector field $B \left( p \right)$ with the four-magnetic field (see e.g., Ref. \cite{Tsamparlis2010} $\S$ 13.10.2), the contravariant components of which are given by,
\begin{equation} \label{Ex:FourMagneticField}
\displaystyle
B^{\mu}_{I} \bigl( \ x_{I} \left( p \right) \ \bigr) = \frac{1}{2} \ \frac{\varepsilon^{\mu \alpha \beta \nu}}{\sqrt{\left( -g \right)}} \ F_{\alpha \beta} \bigl( \ x_{I} \left( p \right) \ \bigr) \ U^{I}_{\nu} \bigl( \ x_{I} \left( p \right) \ \bigr) \\
\end{equation}
\noindent Where $\varepsilon^{\mu \alpha \beta \nu}$ 
is the Levi-Civita tensor density, and $g = \text{det} \Bigl( \ g_{\mu \nu} \bigl( \ x_{I} \left( p \right) \ \bigr) \ \Bigr)$. In the frame of the observer, $F_{\alpha \beta} \bigl( \ x_{I} \left( p \right) \ \bigr)$ is the electromagnetic tensor, and $U^{I}_{\mu} \bigl( \ x_{I} \left( p \right) \ \bigr) = g_{\mu \nu} \bigl( \ x_{I} \left( p \right) \ \bigr) \ U^{\mu}_{I} \bigl( \ x_{I} \left( p \right) \ \bigr)$ are the covariant components of the four-velocity field. 

\end{ex}

\begin{defn} \label{D:MCFlow}
\emph{(The Flow of a Vector Field)}
Let $\mathbb{U} \subseteq \mathbb{M}$ be an open set, and $I \subseteq \mathbb{R}$ and open interval containing $0$. The \textit{flow} of $B$ is a map $\phi : I \times \mathbb{U} \rightarrow \mathbb{M}$, such that for any point $p \in \mathbb{U}$,
\begin{equation} \label{E:MCFlow}
\begin{array}{c}
\displaystyle \frac{\partial}{\partial \lambda} \ \phi \left( \lambda , p \right) = B \bigl( \ \phi \left( \lambda , p \right) \ \bigr) \\ 
\\
\displaystyle
\phi \left( 0 , p \right) = p \\
\end{array}
\end{equation}
\end{defn}

We refer to $B$ as the \textit{generating vector field}, $\lambda$ as the \textit{connectivity parameter} (associated with the generating vector field), $p$ as the \textit{reference point}, and $\phi \left( 0 , p \right) = p$ as the \textit{reference condition}.

If $p \in \mathbb{U}_{I}$ for some coordinate chart $\mathbb{U}_{I}$, then so also lies a segment of $\phi \left( \lambda , p \right)$ in $\mathbb{U}_{I}$. For that segment the coordinate components of equation (\ref{E:MCFlow}) with respect to the basis $\hat{e}^{I}_{\mu}$ are,
\begin{equation} \label{E:MCFlowCoords}
\begin{array}{c}
\displaystyle \frac{\partial}{\partial \lambda} \ x^{\mu}_{I} \bigl( \ \phi \left( \lambda , p \right) \ \bigr) = B^{\mu}_{I} \Bigl( \ x_{I} \bigl( \ \phi \left( \lambda , p \right) \ \bigr) \ \Bigr) \\
\\
\displaystyle
x^{\mu}_{I} \bigl( \ \phi \left( 0 , p \right) \ \bigr) = x^{\mu}_{I} \left( p \right) \\
\end{array}
\end{equation}

Equations (\ref{E:MCFlowCoords}) represent an initial value problem in (four) first-order differential equations in the parameter $\lambda$, covariant under a change of coordinates. Depending on the nature of the generating vector field, equations (\ref{E:MCFlowCoords}) may be non-autonomous or autonomous in the parameter $\lambda$. Hence, for fixed reference condition, by standard theorems of existence, uniqueness, and extension for ordinary differential equations (see e.g., Refs. \cite{HirschSmale1974,Arnold1992,Taylor1996}), the solution $x_{I} \bigl( \ \phi \left( \lambda , p \right) \ \bigr)$ exists, is unique, smooth, and maximal for any $p \in \mathbb{U}_{I}$ and all $\lambda \in I$. 

\begin{ex}
For a generating vector field identified with the four-velocity field (\ref{Ex:FourVelocity}), the associated connectivity parameter is identified with time, and equations (\ref{E:MCFlowCoords}) are non-autonomous.
\end{ex}

\begin{ex}
For a generating vector field identified with the four-magnetic field (\ref{Ex:FourMagneticField}), the associated connectivity parameter is identified with a distance per magnetic field strength, and equations (\ref{E:MCFlowCoords}) are autonomous.
\end{ex}

\begin{rem}
By equation (\ref{E:MCFlowCoords}) the units of the flow $\phi \left( \lambda , p \right)$ are identified with the position coordinates $x^{\mu}_{I} \bigl( \ \phi \left( \lambda , p \right) \ \bigr) \in \mathbb{U}_{I}$. Hence, the units of the connectivity parameter $\lambda$ are coordinate-dependent; that is, in a given coordinate chart, $\lambda$ has coordinate units per generating vector field units. In $\S$ \ref{SS:RefGradDynamics}, we return to a full discussion of the coordinate representation of the connectivity parameter.
\end{rem}

From the flow may be defined two collections of maps: 
\begin{defn} \label{D:MCFlow}
\emph{(The Orbit of $p$)}
For fixed reference point $p \in \mathbb{M}$, the \textit{orbit of $p$} is the $C^{\infty}$ map $\lambda \mapsto \phi_{p} \left( \lambda \right)$ from an interval in $\mathbb{R}$ into $\mathbb{M}$; e.g., $\phi_{p} : I \rightarrow \mathbb{M}$ such that,
\begin{equation} \label{E:MFL}
\displaystyle
\phi_{p} \left( \lambda \right) = \lbrace \ \phi \left( \lambda , p \right) \ \vert \ \lambda \in I \ \text{and fixed} \ p \in \mathbb{M} \ \rbrace
\end{equation}
\end{defn}

\noindent If $p \in \mathbb{U}_{I}$, then the coordinate maps, $\lambda \mapsto x_{I} \bigl( \ \phi_{p} \left( \lambda \right) \ \bigr) \in \mathbb{U}_{I}$ are smooth curves with tangent vector everywhere defined by, and equal to, the (smooth) generating vector field.

\begin{ex}
For a generating vector field identified with the four-velocity field (\ref{Ex:FourVelocity}) in the non-relativistic limit, the orbit of $p$ is a \textit{streamline}.
\end{ex}

\begin{ex}
For a generating vector field identified with the four-magnetic field (\ref{Ex:FourMagneticField}) in the non-relativistic limit, the orbit of $p$ is a \textit{magnetic line of force} (or \textit{magnetic field line}). 
\end{ex}

\begin{defn} \label{D:MCMap}
\emph{(The Connectivity Map)}
For fixed connectivity parameter $\lambda \in I \subseteq \mathbb{R}$, the \textit{connectivity map} is the one-parameter group of (active) diffeomorphisms $\phi_{\lambda} : \mathbb{M} \rightarrow \mathbb{M}$,
\begin{equation} \label{E:MCM}
\displaystyle
\phi_{\lambda} \left( p \right) = \lbrace \ \phi \left( \lambda , p \right) \ \vert \ p \in \mathbb{M} \ \text{and fixed} \ \lambda \in I \ \rbrace
\end{equation}
\end{defn}

If $I = \mathbb{R}$, then for all possible values of connectivity parameter $\lambda \in \mathbb{R}$, the flow $ \phi \left( \lambda , p \right)$ solutions form a Lie group $G$. That is to say, for each value of $\lambda \in \mathbb{R}$ is associated a smooth transformation of the space to itself $p \mapsto \phi \left( \lambda , p \right)$, such that the connectivity map satisfies the following group properties: $\phi_{\lambda_1 + \lambda_2} \left( p \right) = \phi_{\lambda_1} \circ \phi_{\lambda_2} \left( p \right)$ for any $\lambda_1 , \lambda_2 \in \mathbb{R}$; $\phi_{0} \left( p \right) = p$ is the identity element; and $\phi_{-\lambda} \left( p \right)$ is the inverse element, such that $\phi_{-\lambda} \circ \phi_{\lambda} \left( p \right) = \phi_{0} \left( p \right)$. 

\begin{ex}
For divergence-free generating vector fields, the flow is identified with the group of volume-preserving diffeomorphisms $S\text{Diff} \left( \mathbb{M} \right)$ (see e.g., Ref.~\cite{ArnoldKhesin1998} $\S$ I.1).
\end{ex}

The flow represents an equivalence class under a affine transformations of the connectivity parameter, $\lambda \mapsto \lambda \left( \sigma \right)$ with $\lambda \left( 0 \right) = 0$ and $\frac{d \lambda}{d \sigma} > 0$ (see e.g., Ref. \cite{Arnold1992} $\S$ 7.4). In particular, the connectivity parameter may be identified with an arc length measure $s$ along the orbit $\psi_{p} \left( s \right) = \psi \left( s , p \right)$ from the reference point $p \in \mathbb{M}$, where $\phi \left( \lambda , p \right) \mapsto \psi \left( s , p \right)$ by a change of variable $\lambda \mapsto \lambda \left( s \right)$ with, 
\begin{equation} \label{E:ArcLengthCondition}
\begin{array}{c}
\displaystyle
\frac{d \lambda}{d s} = \Big \vert \ B \bigl( \ \psi \left( s , p \right) \ \bigr) \ \Big \vert^{-1} \\
\\
\displaystyle 
\lambda \left( 0 \right) = 0 \\
\end{array}
\end{equation}

For a smooth generating vector field $B$, provided $B \left( p \right) \ne 0$ for all $p \in \mathbb{M}$, it can be shown the transformation (\ref{E:ArcLengthCondition}) exists and is unique. Furthermore, under transformation (\ref{E:ArcLengthCondition}), the generating vector field for the re-parametrized flow $\psi$, is the unit direction of the generating vector field $b \equiv \frac{B}{\vert B \vert}$, such that, 
\begin{equation} \label{E:ArcLengthMCFlow}
\begin{array}{c}
\displaystyle \frac{\partial}{\partial s} \ \psi \left( s , p \right) = b \bigl( \ \psi \left( s , p \right) \ \bigr) \\ 
\\ 
\displaystyle 
\psi \left( 0 , p \right) = p \\
\end{array}
\end{equation}

\noindent Additionally, if $p \in \mathbb{U}_{I}$, then equation (\ref{E:ArcLengthMCFlow}) finds a coordinate expression similar to that of (\ref{E:MCFlow}),
\begin{equation} \label{E:ArcLengthMCFlowCoords}
\begin{array}{c c c}
\displaystyle \frac{\partial}{\partial s} \ x^{\mu}_{I} \bigl( \ \psi \left( s , p \right) \ \bigr) = b^{\mu}_{I} \Bigl( \ x_{I} \bigl( \ \psi \left( s , p \right) \ \bigr) \ \Bigr) \\
\\ 
\displaystyle 
x^{\mu}_{I} \bigl( \ \psi \left( 0 , p \right) \ \bigr) = x^{\mu}_{I} \left( p \right) \\
\end{array}
\end{equation}

We call the re-parametrized flow $\psi \left( s , p \right)$ of equations (\ref{E:ArcLengthMCFlow}), the \textit{arc length representation}, reflecting the fact that the connectivity parameter itself is identified with an arc length measure along the orbit issuing from the reference point $p \in \mathbb{M}$.

\begin{ex}
In the arc length representation of the flow for a generating vector field identified with the four-velocity field (\ref{Ex:FourVelocity}), the connectivity parameter is a measure of the total time along the streamline.
\end{ex}

\begin{ex}
In the arc length representation of the flow for a generating vector field identified with the four-magnetic field (\ref{Ex:FourMagneticField}) in the non-relativistic limit, the connectivity parameter is a measure of the total spatial distance along the magnetic line of force.
\end{ex}

\begin{rem}
Unless explicitly noted we will work with the flow representation $\phi \left( \lambda , p \right)$ associated with the full generating vector field $B$. For completeness though, the flow representations $\phi \left( \lambda , p \right)$ and $\psi \left( s , p \right)$ are implicitly related everywhere via the unit generating field $b \bigl( \ \phi \left( \lambda , p \right) \ \bigr) = b \bigl( \ \psi \left( s , p \right) \ \bigr)$,
\begin{equation} \label{E:FlowArcLengthFlowRelation}
\begin{array}{c}
\displaystyle 
\biggl( \ \frac{\partial}{\partial \lambda} \ \phi \left( \lambda , p \right) \cdot \frac{\partial}{\partial \lambda} \ \phi \left( \lambda , p \right) \biggr)^{-1/2} \ \frac{\partial}{\partial \lambda} \ \phi \left( \lambda , p \right) = \frac{\partial}{\partial s} \ \psi \left( s , p \right) \\
\\
\displaystyle
\phi \left( 0 , p \right) = \psi \left( 0 , p \right) = p \\
\end{array}
\end{equation}

In a coordinate chart $\mathbb{U}_{I}$, equation (\ref{E:FlowArcLengthFlowRelation}) is,
\begin{equation} \label{E:FlowArcLengthFlowRelationCoords}
\begin{array}{c}
\begin{split}
\displaystyle 
\biggl( \ g_{\alpha \beta} \ \frac{\partial}{\partial \lambda} \ x^{\alpha}_{I} \bigl( \ \phi \left( \lambda , p \right) \ \bigr) \ \frac{\partial}{\partial \lambda} \ x^{\beta}_{I} \bigl( \ \phi \left( \lambda , p \right) \ \bigr) &\ \biggr)^{-1/2} \ \frac{\partial}{\partial \lambda} \ x^{\mu}_{I} \bigl( \ \phi \left( \lambda , p \right) \ \bigr) \\
\displaystyle
&= \frac{\partial}{\partial s} \ x^{\mu}_{I} \bigl( \ \psi \left( s , p \right) \ \bigr) \\
\\
\end{split}
\\
\displaystyle
x^{\mu}_{I} \bigl( \ \phi \left( 0 , p \right) \ \bigr) = x^{\mu}_{I} \bigl( \ \psi \left( 0 , p \right) \ \bigr) = x^{\mu}_{I}\left( 0 , p \right) \\
\end{array}
\end{equation}

\noindent Where $g_{\alpha \beta} = g_{\alpha \beta} \Bigl( \ x_{I} \bigl( \ \phi \left( \lambda , p \right) \ \bigr) \ \Bigr)$ is the metric given in the coordinate chart $\mathbb{U}_{I} \subseteq \mathbb{M}$.
\end{rem}

\section{The Referential Gradient of the Flow}
\label{S:RefGradDefined}

Given a generating vector field $B$, the determination of the flow $\phi \left( \lambda , p \right)$ with a particular reference condition $\phi \left( 0 , p \right) = p$, represents an initial value problem, via equations (\ref{E:MCFlow}). The geometric flow structure may be discerned by examining the dependence on the reference condition (see e.g., Ref. \cite{Arnold1992} $\S$ 32).

Let $\mathcal{B}_{r} \left( p \right) \subset \mathbb{M}$ be an open ball of radius $r > 0$ about the reference point $p \in \mathbb{M}$, and $I_{r}$ be a bounded open interval containing $0$ and $\lambda$ with $\overline{I}_{r} \subset I$, such that $\phi \left( \lambda , q \right) \in \overline{\mathcal{B}}_{r} \left( p \right)$ for all points $q \in \overline{\mathcal{B}}_{r} \left( p \right)$ and any $\lambda \in \overline{I}_{r}$. Denote the non-zero \textit{reference shift vector} $h = \vert h \vert \ \hat{h}$ such that $q = p + \vert h \vert \ \hat{h}$, where the unit direction vector $\hat{h} \in \mathcal{B}_{r} \left( 0 \right) \subset T_{p} \mathbb{M}$.

\begin{defn} \label{D:RefGrad}
\emph{(The Referential Gradient of the Flow)}

The \textit{referential gradient of the flow} is the matrix-valued function $F : \overline{I}_{r} \times \overline{\mathcal{B}_{r}} \rightarrow GL \left( 4 , \mathbb{R} \right)$ defined via the vector,
\begin{equation} \label{E:FormalRefGrad}
\displaystyle
F \left( \ \lambda , p \ \right) \cdot \hat{h} = \lim_{\vert h \vert \rightarrow 0} \ \frac{1}{\vert h \vert} \ \biggl( \ \phi \bigl( \lambda , p + \vert h \vert \ \hat{h} \bigr) - \phi \bigl( \lambda , p \bigr) \ \biggr) \\
\\
\end{equation}

\noindent provided the limit exists for arbitrary reference shift vector $\hat{h} \in \mathcal{B}_{r} \left( 0 \right) \subset T_{p} \mathbb{M}$.
\end{defn}

The vector $F \left( \ \lambda , p \ \right) \cdot \hat{h}$ is the referential gradient of the flow corresponding to the direction of the reference shift vector $h$. It follows immediately for $\lambda = 0$, definition (\ref{E:FormalRefGrad}) reduces to $F \left( \ 0 , p \ \right) \cdot \hat{h} = \hat{h}$, thus $F \left( \ 0 , p \ \right) = I$.

\begin{rem}
A somewhat weaker form of the definition of the referential gradient for any $\lambda \in I_{r}$ is such that for any  reference shift vector $h \in T_{p} \mathbb{M}$ and $\tau > 0$,
\begin{equation} \label{E:WeakRefGrad}
\displaystyle
F \left( \ \lambda , p \ \right) \cdot \hat{h} = \frac{d}{d\tau} \Big \vert_{\tau = 0} \ \phi \bigl( \lambda , p + \tau h \bigr)
\end{equation}
%
\end{rem}

Geometrically, for every $\lambda \in I_{r}$, the referential gradient $F \left( \ \lambda , p \ \right)$ is a measure of the relative change of the flow $\phi \left( \lambda , p \right)$ with respect to a shift in the reference point $p \mapsto p + h$. 
In other words, for a given orbit $\phi_{p} \left( \lambda \right) = \phi \left( \lambda , p \right)$, the referential gradient $F \left( \ \lambda , p \ \right)$ contains the deformation information of all ``neighboring'' orbits $\phi_{p + h} \left( \lambda \right) = \phi \left( \lambda , p + h \right)$ of similar length $\lambda \in I_{r} \subset \mathbb{R}$.

Let $\mathcal{B}_{r} \left( p \right) \subseteq \mathbb{U}_{I}$ be in a coordinate chart. The  reference shift direction vector $\hat{h} \in \mathcal{B}_{r} \left( 0 \right) \subset T_{p} \mathbb{U}_{I}$ is given by $\hat{h} = \frac{h^{\mu}_{I}}{\vert h \vert} \ \hat{e}^{I}_{\mu}$, and the coordinate expression for $F \left( \ \lambda , p \ \right) \cdot \hat{h}$ is,
\begin{equation} \label{E:FormalRefGradCoords}
\begin{split}
\displaystyle
F^{\mu}_{\nu} \bigl( \ \lambda , x_{I} &\left( p \right) \ \bigr) \ \hat{h}^{\nu}_{I} \ \hat{e}^{I}_{\mu} \\
\displaystyle
&= \lim_{\vert h \vert \rightarrow 0} \ \frac{1}{\vert h \vert} \ \biggl( \ x^{\mu}_{I} \Bigl( \ \phi \bigl( \lambda , p + \vert h \vert \ \hat{h} \bigr) \ \Bigr) - x^{\mu}_{I} \Bigl( \ \phi \bigl( \lambda , p \bigr) \ \Bigr) \ \biggr) \ \hat{e}^{I}_{\mu} \\
\end{split}
\end{equation}

\noindent provided the limit exists. When $\lambda = 0$, the referential gradient corresponding to the reference shift direction vector $h$ is $F^{\mu}_{\nu} \bigl( \ 0 , x_{I} \left( p \right) \ \bigr) \ \hat{h}^{\nu}_{I} \ \hat{e}^{I}_{\mu} = \hat{h}^{\mu}_{I} \ \hat{e}^{I}_{\mu}$, hence $F^{\mu}_{\nu} \bigl( \ 0 , x_{I} \left( p \right) \ \bigr) = \delta^{\mu}_{\nu}$. 

\subsection{Generalized Lagrangian Specification of the Referential Gradient}
\label{SS:LagrangeSpecRefGrad}

One may always construct a coordinate representation of the referential gradient from the coordinate representation of the flow via definition (\ref{E:FormalRefGradCoords}). However, in typical physical systems of interest, the governing equations describe the evolution of the generating vector field, and the associated flow is derived therefrom. Hence, the following theorem gives the explicit functional dependence of the referential gradient of the flow on the generating vector field. 

\begin{thm}
\emph{(Generalized Lagrangian Specification of the Referential Gradient)} \label{Thm:RefGradLagrangian}
Let $B : \mathbb{M} \rightarrow T \mathbb{M}$ be a smooth, complete generating vector field with associated flow $\phi : \mathbb{R} \times \mathbb{M} \rightarrow \mathbb{M}$. 
The referential gradient $F : \overline{I}_{r} \times \overline{\mathcal{B}_{r}} \rightarrow GL \left( 4 , \mathbb{R} \right)$, for $\mathcal{B}_{r} \subseteq \mathbb{M}$ and $I_{r} \subseteq \mathbb{R}$, is the unique solution to the differential equation,
\begin{equation} \label{E:RefGradHolonomy}
\begin{array}{c}
\displaystyle 
\frac{\partial F \left( \ \lambda , p \ \right)}{\partial \lambda} = \nabla_{F \left( \lambda , p \right)} \ B \Bigl( \ \phi \bigl( \lambda , p \bigr) \ \Bigr) \\ 
\\
\displaystyle
F \left( \ 0 , p \ \right) = I \\
\\
\end{array}
\end{equation}

\noindent Where $I$ is the identity. 
Furthermore, in $\mathcal{B}_{r} \subseteq \mathbb{M}$, 
the solution is given explicitly by the absolutely and uniformly convergent, path-ordered exponential,
\begin{equation} \label{A:PexpReferentialGradient}
\displaystyle
F \left( \ \lambda , p \ \right) = \mathcal{P} \text{exp} \ \int_{0}^{\lambda} d\sigma \ \nabla B \bigl( \ \phi \left( \sigma , p \right) \ \bigr) \\
\\
\end{equation}
\end{thm}

\begin{rem}
The system of differential equations (\ref{E:MCFlow}) and (\ref{E:RefGradHolonomy}) together with their reference conditions are often referred to as the \textit{system of equations of variations} for the flow (see e.g., Ref. \cite{Arnold1992} $\S$ 32). Moreover, equations (\ref{E:RefGradHolonomy}) and (\ref{A:PexpReferentialGradient}) constitute a generalized \textit{Lagrangian specification of the referential gradient} (with respect to the connectivity parameter). 
In $\S$ \ref{SS:RefGradDynamics}, we construct an \textit{Eulerian specification of the referential gradient} equivalent to equations (\ref{E:RefGradHolonomy}).
\end{rem}


\begin{proof}
We begin by deriving equation (\ref{E:RefGradHolonomy}).

%

Using the fundamental relation between the connection flow and generating field, equations (\ref{E:MCFlow}), for any $\lambda \in \overline{I}_{r}$, we construct the flow difference,
\begin{equation} \label{E:MCFlowDifference}
\displaystyle 
\frac{\partial}{\partial \lambda} \biggl( \ \phi \bigl( \lambda , p + \vert h \vert \ \hat{h} \bigr) - \phi \bigl( \lambda , p \bigr) \ \biggr) = B \Bigl( \ \phi \bigl( \lambda , p + \vert h \vert \ \hat{h} \bigr) \ \Bigr) - B \Bigl( \ \phi \bigl( \lambda , p \bigr) \ \Bigr) \\
\end{equation}

\noindent Equation (\ref{E:RefGradHolonomy}) follows upon dividing equation (\ref{E:MCFlowDifference}) by $\vert h \vert$ and in the limit $\vert h \vert \rightarrow 0$.

Consider the LHS of equation (\ref{E:MCFlowDifference}). The vector $h$ does not depend on the connectivity parameter $\lambda$. Using definition (\ref{E:FormalRefGrad}), it follows immediately,
\begin{equation} \label{E:MCFlowDifferenceLHS}
\displaystyle 
\lim_{\vert h \vert \rightarrow 0} \ \frac{1}{\vert h \vert} \ \frac{\partial}{\partial \lambda} \biggl( \ \phi \bigl( \lambda , p + \vert h \vert \ \hat{h} \bigr) - \phi \bigl( \lambda , p \bigr) \ \biggr) = \frac{\partial F \left( \ \lambda , p \ \right)}{\partial \lambda} \cdot \hat{h} \\
\\
\end{equation}

Consider the RHS of equation (\ref{E:MCFlowDifference}). Let $y = \phi \bigl( \lambda , p \bigr) \in \overline{\mathcal{B}}_{r} \left( p \right)$ and denote the non-zero shift vector $w = \vert w \vert \ \hat{w}$ such that $y + \vert w \vert \ \hat{w} = \phi \bigl( \lambda , p + \vert h \vert \ \hat{h} \bigr)$, where the shift direction vector $\hat{w} \in T_{y} \mathbb{M}$.

By assumption, the generating vector field $B \left( p \right)$ is smooth for all $p \in \mathbb{M}$, thus we may Taylor expand the RHS of equation (\ref{E:MCFlowDifference}) (see Appendix \ref{A:GeneralizedFirstOrderTaylorExpansion}). Noting the identity $\nabla_{w} = i_{w} \circ \nabla$, the RHS of equation (\ref{E:MCFlowDifference}) may be written,
\begin{equation} \label{E:TaylorGenField}
\begin{split}
\displaystyle
B \bigl( \ y + \vert w \vert \ \hat{w} \ \bigr) & - B \bigl( y \bigr) = i_{w} \circ \nabla B \bigl( y \bigr) \\
\displaystyle
&+ i_{w} \circ \ \int^{1}_{0} ds \ \biggl( \ \nabla B \bigl( \ y + s \ \vert w \vert \ \hat{w} \ \bigr) - \nabla B \bigl( y \bigr) \ \biggr) \\
\end{split}
\end{equation}

\noindent Hence,
\begin{equation} \label{E:MCFlowDifferenceRHS1}
\begin{array}{l}
\begin{split}
\displaystyle
\lim_{\vert h \vert \rightarrow 0} \ \frac{1}{\vert h \vert} &\ \biggl[ \ B \Bigl( \ \phi \bigl( \lambda , p + \vert h \vert \ \hat{h} \bigr) \ \Bigr) - B \Bigl( \ \phi \bigl( \lambda , p \bigr) \ \Bigr) \ \biggr] \\
\displaystyle
&= \lim_{\vert h \vert \rightarrow 0} \ \frac{w}{\vert h \vert} \cdot \nabla B \bigl( y \bigr) \\
\displaystyle
&+ \lim_{\vert h \vert \rightarrow 0} \ \frac{w}{\vert h \vert} \cdot \int^{1}_{0} ds \ \biggl( \ \nabla B \bigl( \ y + s \ \vert w \vert \ \hat{w} \ \bigr) - \nabla B \bigl( y \bigr) \ \biggr) \\
\end{split}
\end{array}
\end{equation}

\noindent Recall, $w = \phi \bigl( \lambda , p + \vert h \vert \ \hat{h} \bigr) - \phi \bigl( \lambda , p \bigr)$; by definition (\ref{E:FormalRefGrad}),
\begin{equation} \label{E:1stTerm}
\displaystyle
\lim_{\vert h \vert \rightarrow 0} \ \frac{w}{\vert h \vert} = F \left( \ \lambda , p \ \right) \cdot \hat{h} \\
\\
\end{equation}

It remains to show the integral remainder goes to zero in the limit $\vert h \vert \rightarrow 0$. It suffices to show $\vert w \vert \rightarrow 0$ in the limit $\vert h \vert \rightarrow 0$. By assumption, the generating vector field is smooth, then by standard theorems of the smoothness of ODE solutions, so also is the associated flow smooth (e.g., Ref. \cite{Lee2006}, Theorem 17.19); thus,
\begin{equation} \label{E:LimitW}
\displaystyle
\lim_{\vert h \vert \rightarrow 0} \vert w \vert = \lim_{\vert h \vert \rightarrow 0} \ \Big \vert \ \phi \bigl( \lambda , p + \vert h \vert \ \hat{h} \bigr) - \phi \bigl( \lambda , p \bigr) \ \Big \vert = \Big \vert \ \phi \bigl( \lambda , p \bigr) - \phi \bigl( \lambda , p \bigr) \ \Big \vert = 0 \\
\end{equation}

\noindent The integral remainder term in equation (\ref{E:MCFlowDifferenceRHS1}) is zero since $\displaystyle \lim_{\vert h \vert \rightarrow 0} \left( \ y + \vert w \vert \ \hat{w} \ \right) = y$, and,
\begin{equation} \label{E:IntergralRemainder}
\begin{split}
\displaystyle
\lim_{\vert h \vert \rightarrow 0} &\ \frac{w}{\vert h \vert} \cdot \int^{1}_{0} ds \ \biggl( \ \nabla B \bigl( \ y + s \ \vert w \vert \ \hat{w} \ \bigr) - \nabla B \bigl( y \bigr) \ \biggr) \ \biggr] \\
\displaystyle 
&= \Bigl( \ F \left( \ \lambda , p \ \right) \cdot \hat{h} \ \Bigr) \cdot \int^{1}_{0} ds \ \biggl( \ \nabla B \bigl( y \bigr) - \nabla B \bigl( y \bigr) \ \biggr) = 0 \\
\end{split}
\end{equation}

\noindent Hence, the RHS of equation (\ref{E:MCFlowDifference}) is,
\begin{equation} \label{E:MCFlowDifferenceRHS}
\begin{split}
\displaystyle
\lim_{\vert h \vert \rightarrow 0} \ \frac{1}{\vert h \vert} \ \biggl[ \ B \Bigl( \ \phi &\bigl( \lambda , p + \vert h \vert \ \hat{h} \bigr) \ \Bigr) - B \Bigl( \ \phi \bigl( \lambda , p \bigr) \ \Bigr) \ \biggr] \\
\displaystyle
&= \Bigl( \ F \left( \ \lambda , p \ \right) \cdot \hat{h} \ \Bigr) \cdot \nabla B \Bigl( \ \phi \bigl( \lambda , p \bigr) \ \Bigr) \\
\end{split}
\end{equation}

Equating the LHS with the RHS, respectively, equations (\ref{E:MCFlowDifferenceLHS}) and (\ref{E:MCFlowDifferenceRHS}),
\begin{equation} \label{E:DirectionalRefGradHolonomyProof}
\displaystyle
\frac{\partial}{\partial \lambda} \ \Bigl( \ F \left( \ \lambda , p \ \right) \cdot \hat{h} \ \Bigr) = \Bigl( \ F \left( \ \lambda , p \ \right) \cdot \hat{h} \ \Bigr) \cdot \nabla B \Bigl( \ \phi \bigl( \lambda , p \bigr) \ \Bigr) \\
\end{equation}

\noindent Finally, since the reference shift vector $h$ is arbitrary, and noting the contraction map $\nabla_{w} = i_{w} \circ \nabla$ remains unambiguous for $w \in T^{1}_{1} \left( \mathbb{M} \right)$ fields, the differential equation (\ref{E:RefGradHolonomy}) follows immediately,
\begin{equation} \nonumber 
\displaystyle
\frac{\partial F \left( \ \lambda , p \ \right)}{\partial \lambda} = \nabla_{F \left( \lambda , p \right)} \ B \Bigl( \ \phi \bigl( \lambda , p \bigr) \ \Bigr) \\
\\
\end{equation}

\noindent The reference condition $F \left( \ 0 , p \ \right) = I$ is given by definition (\ref{E:FormalRefGrad}) at $\lambda = 0$.

We proceed to construct the solution (\ref{A:PexpReferentialGradient}).

By assumption, the generating vector field $B \left( p \right)$ is smooth and complete for all $p \in \mathbb{M}$, and so therefore the gradient $\nabla B \left( p \right)$ is smooth and everywhere non-singular for all $p \in \mathbb{M}$. Thus, the system of equations (\ref{E:RefGradHolonomy}) is a well-posed initial value problem and standard theorems for existence, uniqueness, and extension for ordinary differential equations apply (see e.g., Refs. \cite{HirschSmale1974,Arnold1992,Taylor1996}). 

To construct solution (\ref{A:PexpReferentialGradient}), we consider the integral formulation of the initial value problem (\ref{E:RefGradHolonomy}); by the fundamental theorem of calculus (see e.g., Ref. \cite{Arnold1992} $\S$ 1.4), equation (\ref{E:RefGradHolonomy}) is equivalent to the integral equation,
\begin{equation} \label{E:IntegralNearbyFLLinearization}
\displaystyle
F \left( \ \lambda , p \ \right) = F \left( \ 0 , p \ \right) + \int_{0}^{\lambda} d\sigma \ F \left( \ \sigma , p \ \right) \cdot \nabla B \Bigl( \ \phi \bigl( \sigma , p \bigr) \ \Bigr)
\end{equation}

\noindent Where $F \left( \ 0 , p \ \right) = I$ is the identity.

\begin{rem}
For fixed $\lambda \in I_{r} \subset \mathbb{R}$, equation (\ref{E:IntegralNearbyFLLinearization}) is a linear Fredholm equation of the second kind. Relaxing the fixed $\lambda$ condition, allowing the upper-limit of integration to vary over all $\lambda \in \mathbb{R}$, equation (\ref{E:IntegralNearbyFLLinearization}) is a linear Volterra equation of the second kind. These integral equations are well known, and the convergence, existence, and uniqueness properties of the solutions are well defined (see e.g., Ref. \cite{Pogorzelski1966}). 
\end{rem}

We proceed to show the existence of a solution to equations (\ref{E:IntegralNearbyFLLinearization}) via the resolvent kernel method. For a non-singular kernel $\nabla B \bigl( \ \phi \left( \lambda , p \right) \ \bigr)$, equation (\ref{E:IntegralNearbyFLLinearization}) may be solved directly by a method of iteration, such that if there exists a set of functions $F \left( \ \lambda , p \ \right)$ that satisfy equations (\ref{E:IntegralNearbyFLLinearization}), then these functions also satisfy the iterated set of equations. Substituting the RHS of equation (\ref{E:IntegralNearbyFLLinearization}) into the integrand and expanding, we get the two-fold iteration equation,
\begin{equation} \label{A:ImplicitDefGrad1}
\begin{array}{l}
\begin{split}
\displaystyle
&F \left( \ \lambda , p \ \right) = F \left( \ 0 , p \ \right) + \int_{0}^{\lambda} d\sigma_{0} \ F \left( \ 0 , p \ \right) \cdot \nabla B \Bigl( \ \phi \bigl( \sigma_{0} , p \bigr) \ \Bigr) \\
\displaystyle
&+ \int_{0}^{\lambda} d\sigma_{0} \ \int_{0}^{\sigma_{0}} d\sigma_{1} \ F \left( \ \sigma_{1} , p \ \right) \cdot \nabla B \Bigl( \ \phi \bigl( \sigma_{1} , p \bigr) \ \Bigr) \cdot \nabla B \Bigl( \ \phi \bigl( \sigma_{0} , p \bigr) \ \Bigr) \\
\end{split}
\end{array}
\end{equation}

\noindent For ease of notation, considering $p$ as a fixed parameter, we define,
\begin{equation} \label{A:Order2IterationKernels}
\displaystyle
N^{\left( 0 \right)} \left( \sigma \right) \equiv \nabla B \Bigl( \ \phi \bigl( \sigma , p \bigr) \ \Bigr) \\
\\
\end{equation}

\noindent Hence, we may write the 2-fold iteration equation (\ref{A:ImplicitDefGrad1}) as,
\begin{equation} \label{A:2FoldIteration}
\begin{split}
\displaystyle
F \left( \ \lambda , p \ \right) &= F \left( \ 0 , p \ \right) + F \left( \ 0 , p \ \right) \cdot \int_{0}^{\lambda} d\sigma_{0} \ N^{\left( 0 \right)} \left( \sigma_{0} \right) \\
\displaystyle
&+ \int_{0}^{\lambda} d\sigma_{0} \int_{0}^{\sigma_{0}} d\sigma_{1} \ F \left( \ \sigma_{1} , p \ \right) \cdot N^{\left( 0 \right)} \left( \sigma_{1} \right) \cdot N^{\left( 0 \right)} \left( \sigma_{0} \right) \\
\end{split}
\end{equation}

\noindent Repeating this procedure $n$-times, and collecting terms, we obtain the $n$-fold iterated equation,
\begin{equation} \label{A:NFoldIteration}
\begin{array}{l}
\displaystyle
F \left( \ \lambda , p \ \right) = F \left( \ 0 , p \ \right) \\
\displaystyle
+ \ F \left( \ 0 , p \ \right) \cdot \int_{0}^{\lambda} d\sigma_{0} \ \biggl[ \ N^{\left( 0 \right)} \left( \sigma_{0} \right) + N^{\left( 1 \right)} \left( \sigma_{0} \right) + \dotsi + N^{\left( n \right)} \left( \sigma_{0} \right) \ \biggr] \\
\begin{split}
\displaystyle
+ \int_{0}^{\lambda} d\sigma_{0} \int_{0}^{\sigma_{0}} d\sigma_{1} \ \dotsi \int_{0}^{\sigma_{n}} d\sigma_{n+1} \ F &\left( \ \sigma_{n+1} , p \ \right) \cdot N^{\left( 0 \right)} \left( \sigma_{n+1} \right) \\
\displaystyle
& \dotsi N^{\left( 0 \right)} \left( \sigma_{1} \right) \cdot N^{\left( 0 \right)} \left( \sigma_{0} \right) \\
\end{split}
\end{array}
\end{equation}

\noindent Where the $n^{th}$ term ($n \ge 1$) is defined by the recursive formula,
\begin{equation} \label{A:NthTermRecurrsionRelation}
\displaystyle
N^{\left( n \right)} \left( \sigma \right) \equiv \int_{0}^{\sigma} d\chi \ N^{\left( n-1 \right)} \left( \chi \right) \cdot N^{\left( 0 \right)} \left( \sigma \right) \\
\end{equation}

\noindent Thus, as $n \rightarrow \infty$, there exist a set of functions, $F \left( \ \lambda , p \ \right)$, that satisfy equations (\ref{E:IntegralNearbyFLLinearization}), and are given by the infinite series,
\begin{equation} \label{A:DefGradGeneralSolution}
\displaystyle
F \left( \ \lambda , p \ \right) = F \left( \ 0 , p \ \right) \cdot \biggl[ \ I + \int_{0}^{\lambda} d\sigma \ \sum_{n = 0}^{\infty} \ N^{\left( n \right)} \left( \sigma \right) \ \biggr] \\
\end{equation}

\noindent Where each term $N^{\left( n \right)} \left( \sigma \right)$ is given by the recursion relations (\ref{A:NthTermRecurrsionRelation}) and (\ref{A:Order2IterationKernels}).

We proceed to show the $n$-fold iterated equation (\ref{A:NFoldIteration}) converges absolutely and uniformly for every $n$, and therefore so does formula (\ref{A:DefGradGeneralSolution}) as $n \rightarrow \infty$. 

To show convergence of the first $n$-terms in equation (\ref{A:NFoldIteration}), let $M = \text{sup} \bigl( \ \big \vert \ N^{\left( 0 \right)} \left( \sigma \right) \big \vert \ \bigr)$ for all $\sigma \in I_{r}$; $M$ exists since, by assumption, $B$ is smooth. Then, by the recursive formula (\ref{A:NthTermRecurrsionRelation}), the absolute value of each successive iterated kernel is also bounded,
\begin{equation} \label{A:BoundedKernels}
\displaystyle
\big \vert \ N^{\left( n \right)} \left( \sigma \right) \big \vert \leq \frac{1}{n!} \ M^{n+1} \ \vert \lambda \vert^{n} \\
\end{equation}

\noindent Where $\vert \lambda \vert$ is the (non-negative) total length of the integration path. 

Formula (\ref{A:BoundedKernels}) is the $n^{th}$ term of the absolutely and uniformly convergent exponential series for any (finite) $\vert \lambda \vert$,
\begin{equation} \label{A:NthKernelConvergence}
\displaystyle
M \ \text{exp} \Bigl[ \ M \ \vert \lambda \vert \ \Bigr] = M \ \sum_{n = 0}^{\infty} \ \frac{1}{n!} \ \Bigl[ \ M \ \vert \lambda \vert \ \Bigr]^{n} 
\end{equation}
\noindent Therefore, the first $n$-terms in equation (\ref{A:NFoldIteration}) are bounded by the absolutely and uniformly convergent series (\ref{A:NthKernelConvergence}) for every $\lambda \in I_{r}$.

To show the final iteration term is bounded, let $m = \text{sup} \Bigl( \ \big \vert \ F \left( \ \sigma , p \ \right) \big \vert \ \Bigr)$ for all $\sigma \in I_{r}$ be the absolute upper bound of the referential gradient function $F \left( \ \sigma , p \ \right)$ everywhere along the integration path. $m$ exists as long as the gradient of the generating vector field remains smooth and non-singular for any $\sigma \in I_{r}$. Then, by the relations (\ref{A:BoundedKernels}), the final iteration term of the $n$-fold iterated equations (\ref{A:NFoldIteration}) satisfies the inequality,
\begin{equation} \label{A:NthTermIterationConvergence}
\begin{array}{r}
\begin{split}
\displaystyle
\bigg \vert \ \int_{0}^{\lambda} d\sigma_{0} \int_{0}^{\sigma_{0}} d\sigma_{1} \dotsi \int_{0}^{\sigma_{n}} d\sigma_{n+1} \ F &\left( \ \sigma_{n+1} , p \ \right) \cdot \ N^{\left( 0 \right)} \left( \sigma_{n+1} \right) \\
\displaystyle
&\dotsi N^{\left( 0 \right)} \left( \sigma_{1} \right) \cdot N^{\left( 0 \right)} \left( \sigma_{0} \right) \ \bigg \vert \\
\\
\end{split}
\\
\displaystyle
\leq \frac{m}{\left( n+1 \right)!} \ \Bigl[ \ M \ \vert \lambda \vert \ \Bigr]^{n+1} \\
\end{array}
\end{equation}

Formula (\ref{A:NthTermIterationConvergence}) is the $\left( n+1 \right)^{th}$ term of the same absolutely and uniformly convergent exponential series (\ref{A:NthKernelConvergence}) with the coefficient $m$,
\begin{equation} \label{A:IterationSolutionConvergence}
\displaystyle
m \ \text{exp} \Bigl[ \ M \ \vert \lambda \vert \ \Bigr] = m \ \sum_{n = 0}^{\infty} \ \frac{1}{n!} \ \Bigl[ \ M \vert \lambda \vert \ \Bigr]^{n} 
\end{equation}

\noindent Thus, the final term in the $n$-fold iterated equation (\ref{A:NFoldIteration}) tends to zero as $n \rightarrow \infty$, for any finite $\lambda \in I_{r}$. Consequently, the $n$-fold iterated equation (\ref{A:NFoldIteration}) converges absolutely and uniformly to the functions $F \left( \ \lambda , p \ \right)$, given by equations (\ref{A:DefGradGeneralSolution}).

We proceed to show that the functions $F \left( \ \lambda , p \ \right)$, given by the absolutely and uniformly convergent series of equations (\ref{A:DefGradGeneralSolution}), indeed satisfy equations (\ref{E:IntegralNearbyFLLinearization}).

Substituting equations (\ref{A:DefGradGeneralSolution}) and definition (\ref{A:Order2IterationKernels}) into equations (\ref{E:IntegralNearbyFLLinearization}),
\begin{equation} \label{A:SatisfySolution}
\begin{split}
\displaystyle
F \left( \ \lambda , p \ \right) &= F \left( \ 0 , p \ \right) + F \left( \ 0 , p \ \right) \cdot \int_{0}^{\lambda} d\sigma \ N^{\left( 0 \right)} \left( \sigma \right) \\
\displaystyle
&+ F \left( \ 0 , p \ \right) \cdot \int_{0}^{\lambda} d\sigma \int_{0}^{\sigma} d\chi \ \sum_{n = 0}^{\infty} \ N^{\left( n \right)} \left( \chi \right) \cdot N^{\left( 0 \right)} \left( \sigma \right) \\
\end{split}
\end{equation}

\noindent Expanding the third term on the RHS of equation (\ref{A:SatisfySolution}),
\begin{equation} \label{A:ExpandedThirdTerm1}
\begin{array}{l}
\displaystyle
F \left( \ 0 , p \ \right) \cdot \int_{0}^{\lambda} d\sigma \int_{0}^{\sigma} d\chi \ \sum_{n = 0}^{\infty} \ N^{\left( n \right)} \left( \chi \right) \cdot N^{\left( 0 \right)} \left( \sigma \right) \\
\begin{split}
\displaystyle
= F \left( \ 0 , p \ \right) \cdot \int_{0}^{\lambda} d\sigma \int_{0}^{\sigma} d\chi \ \biggl[ \ N^{\left( 0 \right)} &\left( \chi \right) + N^{\left( 1 \right)} \left( \chi \right) + \\
\displaystyle
&\dotsi + N^{\left( n \right)} \left( \chi \right) + \dotsi \ \biggr] \cdot N^{\left( 0 \right)} \left( \sigma \right) \\
\end{split}
\end{array}
\end{equation}

\noindent Noting the recursion relation (\ref{A:NthTermRecurrsionRelation}) for successive terms, (\ref{A:ExpandedThirdTerm1}) becomes,
\begin{equation} \label{A:ExpandedThirdTerm2}
\begin{split}
\displaystyle
&F \left( \ 0 , p \ \right) \cdot \ \int_{0}^{\lambda} d\sigma \int_{0}^{\sigma} d\chi \ \sum_{n = 0}^{\infty} \ N^{\left( n \right)} \left( \chi \right) \cdot N^{\left( 0 \right)} \left( \sigma \right) \\
\displaystyle
& = F \left( \ 0 , p \ \right) \cdot \int_{0}^{\lambda} d\sigma \ \biggl[ \ N^{\left( 1 \right)} \left( \sigma \right) + N^{\left( 2 \right)} \left( \sigma \right) + \dotsi + N^{\left( n+1 \right)} \left( \sigma \right) + \dotsi \ \biggr] \\
\end{split}
\end{equation}

\noindent Substituting (\ref{A:ExpandedThirdTerm2}) into (\ref{A:SatisfySolution}),
\begin{equation}
\begin{split}
\displaystyle
&F \left( \ \lambda , p \ \right) = F \left( \ 0 , p \ \right) + F \left( \ 0 , p \ \right) \cdot \int_{0}^{\lambda} d\sigma \ N^{\left( 0 \right)} \left( \sigma \right) \\
\displaystyle
&+ F \left( \ 0 , p \ \right) \cdot \int_{0}^{\lambda} d\sigma \ \biggl[ \ N^{\left( 1 \right)} \left( \sigma \right) + N^{\left( 2 \right)} \left( \sigma \right) + \dotsi + N^{\left( n+1 \right)} \left( \sigma \right) + \dotsi \ \biggr] \\
\end{split}
\end{equation}

\noindent Collecting terms, we recover equations (\ref{A:DefGradGeneralSolution}),
\begin{equation} \nonumber
\displaystyle
F \left( \ \lambda , p \ \right) = F \left( \ 0 , p \ \right) \cdot \biggl[ \ I + \int_{0}^{\lambda} d\sigma \ \sum_{n = 0}^{\infty} \ N^{\left( n \right)} \left( \sigma \right) \ \biggr] \\
\end{equation}

\noindent Thus, equations (\ref{A:DefGradGeneralSolution}), with recursion relations (\ref{A:NthTermRecurrsionRelation}) and (\ref{A:Order2IterationKernels}), satisfy equations (\ref{E:IntegralNearbyFLLinearization}).

Finally, it remains to be shown the iterated solution (\ref{A:DefGradGeneralSolution}) may be cast in closed form, given by the path-ordered exponential $\mathcal{P} \text{exp}$ solution (\ref{A:PexpReferentialGradient}).

Noting recursion relations (\ref{A:NthTermRecurrsionRelation}) and (\ref{A:Order2IterationKernels}), equation (\ref{A:DefGradGeneralSolution}) may be written,
\begin{equation} \label{A:OInfIteration}
\begin{split} 
\displaystyle
& F \left( \ \lambda , p \ \right) = F \left( \ 0 , p \ \right) \cdot \biggl[ \ I + \int_{0}^{\lambda} d\sigma_{0} \ \nabla B \Bigl( \ \phi \bigl( \sigma_{0} , p \bigr) \ \Bigr) \\
\displaystyle
&+ \sum_{n = 1}^{\infty} \ \int_{0}^{\lambda} d\sigma_{0} \ \dotsi \int_{0}^{\sigma_{n-1}} d\sigma_{n} \ \nabla B \Bigl( \ \phi \bigl( \sigma_{n} , p \bigr) \ \Bigr) \ \dotsi \nabla B \Bigl( \ \phi \bigl( \sigma_{0} , p \bigr) \ \Bigr) \ \biggr] \\
\end{split}
\end{equation}


We introduce the \textit{product ordering operator} $\mathcal{P}$ (also known as the \textit{path-ordered product operator}; see e.g., Ref. \cite{PeskinSchroeder1995} $\S$ 4.2), the action of which is to permute the factors $\nabla B \Bigl( \ \phi \bigl( \sigma_{i} , p \bigr) \ \Bigr)$ composing the kernel of each term in equation (\ref{A:OInfIteration}) such that the integration connectivity parameter values $\sigma_{i}$ appear in order from smallest to largest. Explicitly, for any $\sigma_{i} , \sigma_{j} \in I_{r}$, the path-ordered product operator is,
\begin{equation} \label{E:OrderingOperator}
\begin{array}{l}
\displaystyle
\mathcal{P} \biggl( \ \nabla B \Bigl( \ \phi \bigl( \sigma_{i} , p \bigr) \ \Bigr) \cdot \nabla B \Bigl( \ \phi \bigl( \sigma_{j} , p \bigr) \ \Bigr) \ \biggr) \\
\\
\begin{split}
\displaystyle
= \nabla B \Bigl( &\ \phi \bigl( \sigma_{i} , p \bigr) \ \Bigr) \cdot \nabla B \Bigl( \ \phi \bigl( \sigma_{j} , p \bigr) \ \Bigr) \\
\displaystyle
& + \Theta \left( \ \sigma_{i} - \sigma_{j} \ \right) \ \biggl[ \ \nabla B \Bigl( \ \phi \bigl( \sigma_{j} , p \bigr) \ \Bigr) \ , \ \nabla B \Bigl( \ \phi \bigl( \sigma_{i} , p \bigr) \ \Bigr) \ \biggr] \\
%
%
\end{split}
\end{array}
\end{equation}

\noindent Where $\Theta \left( x \right) = \int_{-\infty}^{x} dy \ \delta \left( y \right)$ is the Heavyside step function, with $\Theta \left( 0 \right) \equiv \frac{1}{2}$.

\begin{rem}
In general, the (free) Lie algebra $\mathfrak{g}$ consisting of $\nabla B \bigl( \ \phi \left( \sigma , p \right) \ \bigr) \in \mathfrak{g}$ evaluated at different connectivity parameter values is non-abelian; e.g., for any fixed $\sigma_{i} , \sigma_{j} \in I_{r}$ and $\sigma_{i} \ne \sigma_{j}$,
\begin{equation} \label{E:POCommutator}
\displaystyle
\biggl[ \ \nabla B \Bigl( \ \phi \bigl( \sigma_{j} , p \bigr) \ \Bigr) \ , \ \nabla B \Bigl( \ \phi \bigl( \sigma_{i} , p \bigr) \ \Bigr) \ \biggr] \ne 0 \\
\end{equation}

\noindent However, in the special cases that the (\ref{E:POCommutator}) equals zero, the ordered product operator $\mathcal{P}$ reduces to the standard product.
\end{rem}

Transforming integration limits from $\lambda \geq \sigma_{0} \geq \dotsi \ge \sigma_{n-1} \geq 0$, to a uniform $n$-cube $\lambda \geq \sigma_{i} \geq 0$ for every $i = \lbrace 1 , \dotsi , n-1 \rbrace$, and using the product ordering operator, the $n^{th}$ term ($n \ge 1$) integration in equation (\ref{A:OInfIteration}) may be rewritten,
\begin{equation} \label{A:POInfIteration}
\begin{split}
\displaystyle
& \int_{0}^{\lambda} d\sigma_{0} \ \dotsi \int_{0}^{\sigma_{n-1}} d\sigma_{n} \ \nabla B \Bigl( \ \phi \bigl( \sigma_{n} , p \bigr) \ \Bigr) \ \dotsi \nabla B \Bigl( \ \phi \bigl( \sigma_{0} , p \bigr) \ \Bigr) \\
\displaystyle
& = \frac{1}{n!} \ \int_{0}^{\lambda} d\sigma_{0} \ \dotsi \int_{0}^{\lambda} d\sigma_{n} \ \mathcal{P} \biggl( \ \nabla B \Bigl( \ \phi \bigl( \sigma_{n} , p \bigr) \ \Bigr) \ \dotsi \nabla B \Bigl( \ \phi \bigl( \sigma_{0} , p \bigr) \ \Bigr) \ \biggr) \\
\end{split}
\end{equation}
%

\noindent Furthermore, the nested integral on RHS of equation (\ref{A:POInfIteration}) is simply $n$ independent factor integrations. Hence,
\begin{equation} \label{A:POInfIteration2}
\begin{split}
\displaystyle
\int_{0}^{\lambda} d\sigma_{0} \ \dotsi \int_{0}^{\lambda} d\sigma_{n} \ \mathcal{P} \biggl( \ \nabla B \Bigl( &\ \phi \bigl( \sigma_{n} , p \bigr) \ \Bigr) \ \dotsi \nabla B \Bigl( \ \phi \bigl( \sigma_{0} , p \bigr) \ \Bigr) \ \biggr) \\
\displaystyle
&= \mathcal{P} \biggl( \ \int_{0}^{\lambda} d\sigma \ \nabla B \Bigl( \ \phi \bigl( \sigma , p \bigr) \ \Bigr) \ \biggr)^{n} \\
\end{split}
\end{equation}

From equations (\ref{A:OInfIteration}) and (\ref{A:POInfIteration2}), we may formally define the \textit{path ordered exponential function} $\mathcal{P} \text{exp}$, via its series representation,
\begin{equation} \label{A:POExponential}
\displaystyle
\mathcal{P} \text{exp} \biggl( \ \int_{0}^{\lambda} d\sigma \ \nabla B \Bigl( \ \phi \bigl( \sigma , p \bigr) \ \Bigr) \ \biggr) = \sum_{n = 0}^{\infty} \ \frac{1}{n!} \ \mathcal{P} \biggl( \ \int_{0}^{\lambda} d\sigma \ \nabla B \Bigl( \ \phi \bigl( \sigma , p \bigr) \ \Bigr) \ \biggr)^{n} \\
\\
\end{equation}

\begin{rem}
Cases in which the (free) Lie algebra is abelian, i.e., (\ref{E:POCommutator}) is equal to zero for every $\sigma_{i} , \sigma_{j} \in I_{r}$, the path ordered exponential (\ref{A:POExponential}) reduces to the standard exponential.
\end{rem}

Equation (\ref{A:POExponential}) is simply a formally compact form of the bracketed terms in equation (\ref{A:OInfIteration}). Noting $F \left( \ 0 , p \ \right) = I$ is the identity (given by definition (\ref{E:FormalRefGrad}) at $\lambda = 0$), we get the final form of the referential gradient solution (\ref{A:PexpReferentialGradient}),
\begin{equation} \nonumber 
\displaystyle
F \left( \ \lambda , p \ \right) = \mathcal{P} \text{exp} \biggl( \ \int_{0}^{\lambda} d\sigma \ \nabla B \Bigl( \ \phi \bigl( \sigma , p \bigr) \ \Bigr) \ \biggr) \\
\end{equation}

\begin{rem}
For consistency, we note the formal path ordered exponential solution (\ref{A:PexpReferentialGradient}) trivially reduces to the identity at $\lambda = 0$.
\end{rem}

\end{proof}

It is a nearly trivial matter to write the coordinate representation of equations (\ref{E:RefGradHolonomy}) and solution (\ref{A:PexpReferentialGradient}). 
Let $\mathcal{B}_{r} \left( p \right) \subseteq \mathbb{U}_{I}$ be in a coordinate chart. The reference shift vector may be written $h = h^{\mu}_{I} \ \hat{e}^{I}_{\mu}$. Thus, the vector $F \left( \ \lambda , p \ \right) \cdot \hat{h} = F^{\mu}_{\nu} \left( \ \lambda , p \ \right) \ \hat{h}^{\nu}_{I} \ \hat{e}^{I}_{\mu}$ representing the referential gradient of the flow vector corresponding to the reference shift vector $h$, evolves with the connectivity parameter $\lambda \in I_{r}$ according to,
\begin{equation} \label{E:RefGradHolonomyCoords1}
\begin{split}
\displaystyle 
\frac{\partial F^{\mu}_{\nu} \bigl( \ \lambda , x_{I} \left( p \right) \ \bigr)}{\partial \lambda} &\ \hat{h}^{\nu}_{I} \ \hat{e}^{I}_{\mu} \\
\displaystyle
&= F^{\eta}_{\nu} \bigl( \ \lambda , x_{I} \left( p \right) \ \bigr) \ \nabla_{\eta} \ B^{\mu}_{I} \biggl( \ x_{I} \Bigl( \ \phi \bigl( \lambda , p \bigr) \ \Bigr) \ \biggr) \ \hat{h}^{\nu}_{I} \ \hat{e}^{I}_{\mu} \\
\end{split}
\end{equation}

\noindent Since the reference shift vector $h = h^{\mu}_{I} \ \hat{e}^{I}_{\mu}$ is arbitrary, we may write the first-order differential equations (\ref{E:RefGradHolonomy}) for the coordinate-dependent component representation $F^{\mu}_{\nu} \bigl( \ \lambda , x_{I} \left( p \right) \ \bigr)$,
\begin{equation} \label{E:RefGradHolonomyCoords}
\displaystyle 
\frac{\partial F^{\mu}_{\nu} \bigl( \ \lambda , x_{I} \left( p \right) \ \bigr)}{\partial \lambda} = F^{\eta}_{\nu} \bigl( \ \lambda , x_{I} \left( p \right) \ \bigr) \ \nabla_{\eta} \ B^{\mu}_{I} \Bigl( \ x_{I} \bigl( \ \phi \left( \lambda , p \right) \ \bigr) \ \Bigr) \\
\end{equation}

\noindent Where the reference condition $F^{\mu}_{\nu} \bigl( \ \lambda , x_{I} \left( p \right) \ \bigr) = \delta^{\mu}_{\nu}$ is given by definition (\ref{E:FormalRefGrad}) at $\lambda = 0$.
%

Furthermore, in the coordinate chart $\mathbb{U}_{I}$, solution (\ref{A:PexpReferentialGradient}) is written,
\begin{equation} \label{E:RefGradProofCoords}
\displaystyle
F^{\mu}_{\nu} \bigl( \ \lambda , x_{I} \left( p \right) \ \bigr) = \mathcal{P} \text{exp} \biggl( \ \int_{0}^{\lambda} d\sigma \ \nabla_{\nu} \ B^{\mu}_{I} \Bigl( \ x_{I} \bigl( \ \phi \bigl( \sigma , p \bigr) \ \bigr) \ \Bigr) \ \biggr) \\
\\
\end{equation}

\noindent The kernel is the covariant derivative of the generating vector field is given by,
\begin{equation} \label{E:TensorGrad}
\displaystyle
\nabla_{\nu} \ B^{\mu}_{I} \left( \ x_{I} \ \right) = \frac{\partial}{\partial x^{\nu}_{I}} \ B^{\mu}_{I} \left( \ x_{I} \ \right) + \Gamma^{\mu}_{\nu \rho} \left( \ x_{I} \ \right) \ B^{\rho}_{I} \left( \ x_{I} \ \right) \\
\\
\end{equation}

\noindent Where the position $x_{I} = x_{I} \bigl( \ \phi \left( \sigma , p \right) \ \bigr)$ is evaluated along the flow.

\subsection{Referential Gradient Transformations}
\label{SS:RefGradTransforms}

The integration operation of the referential gradient solution (\ref{A:PexpReferentialGradient}) places powerful non-trivial restrictions on this object's transformation properties. In this section, we explore the various conditions demanded in order that: 1) the components of the referential gradient transform as a tensor, 2) the referential gradient forms a group under translations of the connectivity parameter, and 3) the relationship between representations of the referential gradient under a change in integration variable.

It is important to note, due to the integration the coordinate representation of the Lagrangian specification of the referential gradient solution (\ref{E:RefGradProofCoords}) is valid only to the extent that it domain of definition $\mathcal{B}_{r} \left( p \right) \subseteq \mathbb{U}_{I}$. The first lemma makes explicit the condition under which the referential gradient components in any coordinate representation transform as a tensor under smooth changes of coordinates; namely, so long as the set $\mathcal{B}_{r} \left( p \right) \subseteq \mathbb{U}_{I} \bigcap \mathbb{U}_{J}$.

\begin{lem}
\emph{(Coordinate Transformation of the Referential Gradient)}
\label{lem:RefGrad}

Let $\left( \ \mathbb{U}_{I} , x_{I} \ \right)_{I \in \mathcal{I}}$ and $\left( \ \mathbb{U}_{J} , x_{J} \ \right)_{J \in \mathcal{I}}$ be coordinate charts on $\mathbb{M}$. Under smooth coordinate transformations $\Lambda : \mathbb{U}_{J} \rightarrow \mathbb{U}_{I}$, the coordinate representation of the referential gradient transforms as a tensor,
\begin{equation} \label{E:RefGradCoordTransform}
\displaystyle
F^{\mu}_{\nu} \bigl( \ \lambda , x_{I} \left( p \right) \ \bigr) = \Lambda^{\mu}_{\alpha} \ \Lambda^{\beta}_{\nu} \ F^{\alpha}_{\beta} \bigl( \ \lambda , x_{J} \left( p \right) \ \bigr)
\end{equation}

\noindent if and only if $\overline{\mathcal{B}}_{r} \left( p \right) \subseteq \mathbb{U}_{I} \bigcap \mathbb{U}_{J}$.
\end{lem}

\begin{proof}
To prove necessity, we assume $\overline{\mathcal{B}}_{r} \left( p \right) \subseteq \mathbb{U}_{I} \bigcap \mathbb{U}_{J}$, and show the transformation (\ref{E:RefGradCoordTransform}) follows from the series solution (\ref{A:OInfIteration}) for any $\lambda \in I_{r}$.

In the chart $\mathbb{U}_{I}$, the coordinate representation of the referential gradient series solution (\ref{A:OInfIteration}) is given by,
\begin{equation} \label{E:SeriesSolnCoords}
\begin{array}{l}
\displaystyle
F^{\mu}_{\nu} \bigl( \ \lambda , x_{I} \left( p \right) \ \bigr) = \delta^{\mu}_{\nu} + \int_{0}^{\lambda} d\sigma_{0} \ \nabla_{\nu} \ B^{\mu}_{I} \Bigl( \ x_{I} \bigl( \ \phi \bigl( \sigma_{0} , p \bigr) \ \bigr) \ \Bigr) \\
\begin{split}
\displaystyle
+ \sum_{n = 1}^{\infty} \ \int_{0}^{\lambda} d\sigma_{0} \ \dotsi \int_{0}^{\sigma_{n-1}} d\sigma_{n} \ \nabla_{\nu} &\ B^{\rho_{n}}_{I} \Bigl( \ x_{I} \bigl( \ \phi \bigl( \sigma_{n} , p \bigr) \ \bigr) \ \Bigr) \\
\displaystyle
& \dotsi \nabla_{\rho_{1}} \ B^{\mu}_{I} \Bigl( \ x_{I} \bigl( \ \phi \bigl( \sigma_{0} , p \bigr) \ \bigr) \ \Bigr) \\
\end{split}
\end{array}
\end{equation}

\noindent where we have made use of the reference condition $F^{\mu}_{\nu} \bigl( \ 0 , x_{I} \left( p \right) \ \bigr) = \delta^{\mu}_{\nu}$.

By assumption $\overline{\mathcal{B}}_{r} \left( p \right) \subseteq \mathbb{U}_{I} \bigcap \mathbb{U}_{J}$, and each kernel factor covariant derivative of the generating vector field transforms under smooth coordinate transformations as,
\begin{equation} \label{E:KernelTransform}
\displaystyle
\nabla_{\nu} \ B^{\mu}_{I} \Bigl( \ x_{I} \bigl( \ \phi \bigl( \sigma_{i} , p \bigr) \ \bigr) \ \Bigr) = \Lambda^{\beta}_{\nu} \ \Lambda^{\mu}_{\alpha} \ \nabla_{\beta} \ B^{\alpha}_{J} \Bigl( \ x_{J} \bigl( \ \phi \bigl( \sigma_{i} , p \bigr) \ \bigr) \ \Bigr) \\
\end{equation}

\noindent for every $\sigma_{i} \in I_{r}$.

Making use of the identity $\Lambda^{\rho}_{\eta} \ \Lambda^{\delta}_{\rho} = \delta^{\delta}_{\eta}$, the $n = 1$ product kernel transforms as,
\begin{equation} \label{E:2ndOrderProductKernelTransform}
\begin{split}
\displaystyle
\nabla_{\nu} \ & B^{\rho_{1}}_{I} \Bigl( \ x_{I} \bigl( \ \phi \bigl( \sigma_{1} , p \bigr) \ \bigr) \ \Bigr) \ \nabla_{\rho_{1}} \ B^{\mu}_{I} \Bigl( \ x_{I} \bigl( \ \phi \bigl( \sigma_{0} , p \bigr) \ \bigr) \ \Bigr) \\
\displaystyle
& =\Lambda^{\beta}_{\nu} \ \Lambda^{\mu}_{\alpha} \ \nabla_{\beta} \ B^{\eta_{1}}_{J} \Bigl( \ x_{J} \bigl( \ \phi \bigl( \sigma_{1} , p \bigr) \ \bigr) \ \Bigr) \ \nabla_{\eta_{1}} \ B^{\alpha}_{J} \Bigl( \ x_{J} \bigl( \ \phi \bigl( \sigma_{0} , p \bigr) \ \bigr) \ \Bigr) \\
\end{split}
\end{equation}

\noindent By induction, the $n^{th}$ product kernel transforms as,
\begin{equation} \label{E:nthOrderProductKernelTransform}
\begin{split}
\displaystyle
&\nabla_{\nu} \ B^{\rho_{n}}_{I} \Bigl( \ x_{I} \bigl( \ \phi \bigl( \sigma_{n} , p \bigr) \ \bigr) \ \Bigr) \ \dotsi \nabla_{\rho_{1}} \ B^{\mu}_{I} \Bigl( \ x_{I} \bigl( \ \phi \bigl( \sigma_{0} , p \bigr) \ \bigr) \ \Bigr) \\
\displaystyle
& = \Lambda^{\beta}_{\nu} \ \Lambda^{\mu}_{\alpha} \ \nabla_{\beta} \ B^{\eta_{n}}_{J} \Bigl( \ x_{J} \bigl( \ \phi \bigl( \sigma_{n} , p \bigr) \ \bigr) \ \Bigr) \ \dotsi \nabla_{\eta_{1}} \ B^{\alpha}_{J} \Bigl( \ x_{J} \bigl( \ \phi \bigl( \sigma_{0} , p \bigr) \ \bigr) \ \Bigr) \\
\end{split}
\end{equation}

Upon substitution of (\ref{E:KernelTransform}) and (\ref{E:nthOrderProductKernelTransform}) into (\ref{E:SeriesSolnCoords}), and using the fact that the connectivity parameter is independent of coordinates, equation (\ref{E:RefGradCoordTransform}) follows immediately,
\begin{equation} \nonumber
F^{\mu}_{\nu} \bigl( \ \lambda , x_{I} \left( p \right) \ \bigr) =  \Lambda^{\beta}_{\nu} \ \Lambda^{\mu}_{\alpha} \ F^{\alpha}_{\beta} \bigl( \ \lambda , x_{J} \left( p \right) \ \bigr)
\end{equation}

To prove sufficiency, let the chart $\mathbb{U}_{I}$ cover the set $\mathcal{B}_{r} \left( p \right)$, and let $\overline{\mathbb{U}}_{J} \subset \mathbb{U}_{I}$. For any reference point $p \in \mathbb{U}_{I} \bigcap \mathbb{U}_{J}$, the coordinate representation of the referential gradient $F^{\mu}_{\nu} \bigl( \ \lambda , x_{I} \left( p \right) \ \bigr)$ in the chart $\mathbb{U}_{I}$ is well-defined for all $\lambda \in I_{r}$. Furthermore, there exists a $\overline{\lambda} \in I_{r}$ such that $x_{I} \bigl( \ \phi \left( \lambda , p \right) \ \bigr) \in U_{I} \ / \ U_{J}$ for any $\vert \ \lambda \ \vert > \vert \ \overline{\lambda} \ \vert$, thus equation (\ref{E:RefGradCoordTransform}) is not defined. Hence, upon shrinking $I_{r}$ to the extent that $\mathcal{B}_{r} \left( p \right) \subseteq \mathbb{U}_{I} \bigcap \mathbb{U}_{J}$ guarantees equation (\ref{E:RefGradCoordTransform}) for smooth coordinate transformations $\Lambda : \mathbb{U}_{J} \rightarrow \mathbb{U}_{I}$.

\end{proof}
%

Whereas the previous Lemma concerned the tensorial nature of the referential gradient components in any coordinate representation, the second lemma concerns the group properties of the referential gradient with respect to the connectivity parameter $\lambda \in I_{r}$. Namely, due to the integration, the group property of the referential gradient with respect to translations of connectivity parameter is a direct consequence of the group property of the flow along which it is evaluated.

\begin{lem}
\emph{(Group Structure of the Referential Gradient)}
\label{lem:RefGradGroup}
Let $p = \phi \left( 0 , p \right)$ and $q = \phi \left( \lambda_{1} , p \right)$. The referential gradient forms a (Lie) group in the connectivity parameter.
\begin{equation} \label{E:FGroup}
\displaystyle
F \bigl( \ \lambda_{2} + \lambda_{1} , p \ \bigr) = F \bigl( \ \lambda_{2} , q \ \bigr) \cdot F \bigl( \ \lambda_{1} , p \ \bigr) \\
\\
\end{equation}

\noindent if and only if the flow satisfies $\phi \left( \lambda_{2} , q \right) = \phi \left( \lambda_{2} + \lambda_{1} , p \right)$ for every  $\lambda_{1}, \lambda_{2} \in I_{r}$ and $p , q \in \mathbb{M}$, or the representation of $\nabla B$ in any coordinate chart is a constant. 
\end{lem}

\begin{proof} 
To prove the lemma, it suffices to work in coordinates; let $\mathcal{B}_{r} \left( p \right) \subseteq \mathbb{U}_{I}$ be in a coordinate chart. To show necessity, we begin with the defining equation (\ref{E:RefGradHolonomyCoords}).

Without loss of generality, fix $\lambda_{1}$ such that $\lambda = \sigma + \lambda_{1}$. Noting, $\frac{\partial}{\partial \lambda} = \frac{\partial}{\partial \sigma}$, equation (\ref{E:RefGradHolonomyCoords}) becomes,
\begin{equation} \label{E:RefGradHolonomy1}
\begin{split}
\displaystyle 
\frac{\partial}{\partial \sigma} &\ F^{\mu}_{\nu} \bigl( \ \sigma + \lambda_{1} , x_{I} \left( p \right) \ \bigr) \\
\displaystyle
&= F^{\eta}_{\nu} \bigl( \ \sigma + \lambda_{1} , x_{I} \left( p \right) \ \bigr) \ \nabla_{\eta} \ B^{\mu}_{I} \Bigl( \ x_{I} \bigl( \ \phi \left( \sigma + \lambda_{1} , p \right) \ \bigr) \ \Bigr) \\
\end{split}
\end{equation}

\noindent Suppose the referential gradient satisfies the group condition, equation (\ref{E:FGroup}), then we may write,
\begin{equation} \label{E:RefGradHolonomy2}
\begin{array}{l}
\begin{split}
\displaystyle 
&\biggl( \ \frac{\partial}{\partial \sigma} \ F^{\mu}_{\alpha} \bigl( \ \sigma , x_{I} \left( q \right) \ \bigr) \ \biggr) \ F^{\alpha}_{\nu} \bigl( \ \lambda_{1} , x_{I} \left( p \right) \ \bigr) \\
\displaystyle
&= F^{\eta}_{\alpha} \bigl( \ \sigma , x_{I} \left( q \right) \ \bigr) \ F^{\alpha}_{\nu} \bigl( \ \lambda_{1} , x_{I} \left( p \right) \ \bigr) \ \nabla_{\eta} \ B^{\mu}_{I} \Bigl( \ x_{I} \bigl( \ \phi \left( \sigma + \lambda_{1} , p \right) \ \bigr) \ \Bigr) \\
\end{split}
\end{array}
\end{equation}

\noindent Substituting equation (\ref{E:RefGradHolonomyCoords}),
\begin{equation} \label{E:RefGradHolonomy3}
\begin{split}
\displaystyle 
&F^{\eta}_{\alpha} \bigl( \ \sigma , x_{I} \left( q \right) \ \bigr) \ \nabla_{\eta} \ B^{\mu}_{I} \Bigl( \ x_{I} \bigl( \ \phi \left( \sigma , q \right) \ \bigr) \ \Bigr) \ F^{\alpha}_{\nu} \bigl( \ \lambda_{1} , x_{I} \left( p \right) \ \bigr) \\
\displaystyle
&= F^{\eta}_{\alpha} \bigl( \ \sigma , x_{I} \left( q \right) \ \bigr) \ F^{\alpha}_{\nu} \bigl( \ \lambda_{1} , x_{I} \left( p \right) \ \bigr) \ \nabla_{\eta} \ B^{\mu}_{I} \Bigl( \ x_{I} \bigl( \ \phi \left( \sigma + \lambda_{1} , p \right) \ \bigr) \ \Bigr) \\
\end{split}
\end{equation}

\noindent Fixing $\sigma = \lambda_{2}$, we find the necessary condition that the referential gradient is a group in the connectivity parameter follows,
\begin{equation} \label{E:NecessaryCondition}
\displaystyle 
\nabla_{\eta} \ B^{\mu}_{I} \Bigl( \ x_{I} \bigl( \ \phi \left( \lambda_{2} , q \right) \ \bigr) \ \Bigr) = \nabla_{\eta} \ B^{\mu}_{I} \Bigl( \ x_{I} \bigl( \ \phi \left( \lambda_{2} + \lambda_{1} , p \right) \ \bigr) \ \Bigr) \\
\end{equation}

\noindent Equation (\ref{E:NecessaryCondition}) is true for the representation of $\nabla B$ in any coordinate chart if $\phi \left( \lambda_{2} , q \right) = \phi \left( \lambda_{2} + \lambda_{1} , p \right)$ for every  $\lambda_{1}, \lambda_{2} \in I_{r}$ and $p , q \in \mathbb{M}$, or is constant (independent of both the connectivity parameter $\lambda \in I_{r}$ and the reference point $p \in \mathbb{M}$).

To show sufficiency, we make use of the path-ordered exponential functional solution, equation (\ref{E:RefGradProofCoords}). 

For fixed $\lambda_{1} , \lambda_{2}$ such that $\lambda = \lambda_{1} + \lambda_{2} \in I_{r}$,
\begin{equation} \label{E:PexpDecomp}
\displaystyle
F^{\mu}_{\nu} \bigl( \ \lambda_{2} + \lambda_{1} , x_{I} \left( p \right) \ \bigr) = \mathcal{P} \text{exp} \biggl( \ \int_{0}^{\lambda_{1} + \lambda_{2}} d\sigma \ \nabla_{\nu} \ B^{\mu}_{I} \Bigl( \ x_{I} \bigl( \ \phi \left( \sigma , p \right) \ \bigr) \ \Bigr) \ \biggr) \\
\end{equation}

\noindent The integral argument in the path-ordered exponential may be written as a sum,
\begin{equation} \label{E:PexpDecomp}
\begin{array}{l}
\begin{split}
\displaystyle
&\int_{0}^{\lambda_{1} + \lambda_{2}} d\sigma \ \nabla_{\nu} \ B^{\mu}_{I} \Bigl( \ x_{I} \bigl( \ \phi \bigl( \sigma , p \bigr) \ \bigr) \ \Bigr) \\
\displaystyle
&= \int^{\lambda_{2}}_{0} d\tau \ \nabla_{\nu} \ B^{\mu}_{I} \Bigl( \ x_{I} \bigl( \ \phi \left( \tau + \lambda_{1} , p \right) \ \bigr) \ \Bigr) + \int_{0}^{\lambda_{1}} d\sigma \ \nabla_{\nu} \ B^{\mu}_{I} \Bigl( \ x_{I} \bigl( \ \phi \left( \sigma , p \right) \ \bigr) \ \Bigr) \\
\end{split}
\end{array}
\end{equation}

\noindent Where, in the first integral on the RHS we have made the substitution $\tau = \sigma - \lambda_{1}$.

Since the flow $\phi \left( \lambda , p \right)$ forms a group in the connectivity parameter, the second integral on the RHS is equivalent to pivoting off of a different reference point; namely, $x_{I} \left( q \right) = x_{I} \bigl( \ \phi \left( 0 , q \right) \ \bigr) = x_{I} \bigl( \ \phi \left( \lambda_{1} , p \right) \ \bigr)$. Hence, we may write,
\begin{equation} \label{E:PexpDecomp}
\begin{array}{c}
\begin{split}
\displaystyle
F^{\mu}_{\nu} \bigl( \ \lambda_{2} + \lambda_{1} , p \ \bigr) = \mathcal{P} \text{exp} &\biggl( \ \int^{\lambda_{2}}_{0} d\tau \ \nabla_{\nu} \ B^{\mu}_{I} \Bigl( \ x_{I} \bigl( \ \phi \bigl( \tau , q \bigr) \ \bigr) \ \Bigr) \\
\displaystyle
&+ \int_{0}^{\lambda_{1}} d\sigma \ \nabla_{\nu} \ B^{\mu}_{I} \Bigl( \ x_{I} \bigl( \ \phi \bigl( \sigma , p \bigr) \ \bigr) \ \Bigr) \ \biggr) \\
\end{split}
\end{array}
\end{equation}

By the Baker-Campbell-Hausdorff Theorem (see Appendix \ref{A:BCHTheorem}), sufficiency follows from the condition that the (free) Lie algebra is abelian; e.g., the following commutator is zero,
\begin{equation} \label{E:GroupCommutator}
\begin{array}{l}
\displaystyle
\biggl[ \ \int^{\lambda_{2}}_{0} d\tau \ \nabla_{\nu} \ B^{\mu}_{I} \Bigl( \ x_{I} \bigl( \ \phi \bigl( \tau , q \bigr) \ \bigr) \ \Bigr) \ , \ \int_{0}^{\lambda_{1}} d\sigma \ \nabla_{\nu} \ B^{\mu}_{I} \Bigl( \ x_{I} \bigl( \ \phi \bigl( \sigma , p \bigr) \ \bigr) \ \Bigr) \ \biggr] \\
\\
\displaystyle
= \int_{0}^{\lambda_{2}} d\tau \ \int^{\lambda_{1}}_{0} d\sigma \ \biggl[ \ \nabla_{\nu} \ B^{\mu}_{I} \Bigl( \ x_{I} \bigl( \ \phi \bigl( \tau , q \bigr) \ \bigr) \ \Bigr) , \nabla_{\nu} \ B^{\mu}_{I} \Bigl( \ x_{I} \bigl( \ \phi \bigl( \sigma , p \bigr) \ \bigr) \ \Bigr) \ \biggr] \\
\end{array}
\end{equation}

\noindent The RHS of equation (\ref{E:GroupCommutator}) is zero if the kernel is zero, to wit,
\begin{equation} \label{E:GroupCommutatorAbelian}
\displaystyle
\biggl[ \ \nabla_{\nu} \ B^{\mu}_{I} \Bigl( \ x_{I} \bigl( \ \phi \bigl( \tau , q \bigr) \ \bigr) \ \Bigr) \ , \ \nabla_{\nu} \ B^{\mu}_{I} \Bigl( \ x_{I} \bigl( \ \phi \bigl( \sigma , p \bigr) \ \bigr) \ \Bigr) \ \biggr] = 0 \\
\end{equation}

In general, equation (\ref{E:GroupCommutatorAbelian}) is non-zero unless the representation of $\nabla B$ in any coordinate chart is: 1) evaluated at the same point, or 2) independent of the connectivity parameter and reference point.

The second case is trivial. 

The first case requires the condition $\phi \left( \tau , q \right) = \phi \left( \sigma , p \right)$. Substituting for $\sigma$ the change of integration variable used in equation (\ref{E:PexpDecomp}), we recover the condition, $\phi \left( \tau , q \right) = \phi \left( \tau + \lambda_{1} , p \right)$.
%

\end{proof}

Lemma (\ref{lem:RefGradGroup}) simply states that the group property of the referential gradient with respect to translations of the connectivity parameter is a direct consequence of the similar group property of the flow; the constant $\nabla B$ representation is the trivial case.
In a follow-up paper, we explore the geometric and topological implications of this lemma.

The third lemma relates the referential gradient of the respective representations of the flow. Recall from section (\ref{SS:GeneralConnectivityMap}), the flow $\phi \left( \lambda , p \right)$ of a vector field is an equivalence class under affine transformations of the connectivity parameter. As such, the flow may always be represented in a form $\psi \left( l , p \right)$ in which the connectivity parameter measures an arc-length from the reference point  via relations (\ref{E:ArcLengthCondition}), (\ref{E:ArcLengthMCFlow}), and (\ref{E:FlowArcLengthFlowRelation}).

\begin{rem}
For ease of notation, by $F \left[ \ B \ \right]$ and $F \left[ \ b \ \right]$ we mean, respectively, the Lagrangian specification of the referential gradient solution (\ref{A:PexpReferentialGradient}) constructed from of the full generating vector field $B$ with flow representation $\phi \left( \lambda , p \right)$, and that constructed from the unit generating vector field $b$ with arc-length flow representation $\psi \left( l , p \right)$.
\end{rem}

\begin{lem}
\emph{(Relations Between Referential Gradient Representations)}
\label{lem:RepresentationRelations}
Let $B : \mathbb{M} \rightarrow T \mathbb{M}$ be a non-null generating vector field. The representations $F \left[ \ B \ \right]$ and $F \left[ \ b \ \right]$ are related by,
\begin{equation} \label{E:RepresentationsRelatations}
\begin{array}{c}
\displaystyle
F \left[ \ B \ \right] = \mathcal{P} \text{exp} \int^{l}_{0} ds \ \biggl( \ \nabla b \bigl( \ \psi \left( s , p \right) \ \bigr) + \frac{b}{2} \ \frac{\nabla B^{2}}{B^{2}} \bigl( \ \psi \left( s , p \right) \ \bigr) \ \biggr) \\
\\
\displaystyle
F \left[ \ b \ \right] = \mathcal{P} \text{exp} \int^{\lambda}_{0} d\sigma \ \biggl( \ \nabla B \bigl( \ \phi \left( \sigma , p \right) \ \bigr) - \frac{B}{2} \ \frac{\nabla B^{2}}{B^{2}} \bigl( \ \phi \left( \sigma , p \right) \ \bigr) \ \biggr) \\
\end{array}
\end{equation}

\noindent Furthermore, if the (free) Lie algebras are abelian,
\begin{equation} \label{E:Commutators}
\begin{array}{c}
\displaystyle
\biggl[ \ \int^{l}_{0} ds \ \nabla b \bigl( \ \psi \left( s , p \right) \ \bigr) \ , \ \int^{l}_{0} ds \ \frac{b}{2} \ \frac{\nabla B^{2}}{B^{2}} \bigl( \ \psi \left( s , p \right) \ \bigr) \ \biggr] = 0 \\
\\
\displaystyle
\biggl[ \ \int^{\lambda}_{0} d\sigma \ \nabla B \bigl( \ \phi \left( \sigma , p \right) \ \bigr) \ , \ \int^{l}_{0} ds \ \frac{B}{2} \ \frac{\nabla B^{2}}{B^{2}} \bigl( \ \phi \left( \sigma , p \right) \ \bigr) \ \biggr] = 0 \\
\end{array}
\end{equation}

\noindent Then relations (\ref{E:RepresentationsRelatations}) may be written,
\begin{equation} \label{E:RepresentationsRelatationDecomposition}
\begin{array}{c}
\displaystyle
F \left[ \ B \ \right] = F \left[ \ b \ \right] \ \mathcal{P} \text{exp} \biggl( \ \int^{l}_{0} ds \ \frac{b}{2} \ \frac{\nabla B^{2}}{B^{2}} \bigl( \ \psi \left( s , p \right) \ \bigr) \ \biggr) \\
\\
\displaystyle
F \left[ \ b \ \right] = F \left[ \ B \ \right] \ \mathcal{P} \text{exp} \biggl( \ - \int^{\lambda}_{0} d\sigma \ \frac{B}{2} \ \frac{\nabla B^{2}}{B^{2}} \bigl( \ \phi \left( \sigma , p \right) \ \bigr) \ \biggr) \\
\end{array}
\end{equation}
\end{lem}

\begin{rem}
We note, in the case of a null generating vector field, $B^{2} = 0$. Hence, the unit generating direction field $b$ does not exist, and therefore neither does the representation $F \left[ \ b \ \right]$.
\end{rem}

\begin{proof}
We begin by deriving relations (\ref{E:RepresentationsRelatations}). 

By Theorem (\ref{Thm:RefGradLagrangian}), each representation independently satisfies the corresponding system of equations, 
\begin{equation} \label{E:SystemRepresentation}
\begin{array}{c}
\displaystyle 
\frac{\partial}{\partial \lambda} \ F \left[ \ B \ \right] = F \left[ \ B \ \right] \cdot \nabla B \bigl( \ \phi \left( \lambda , p \right) \ \bigr) \\
\\
\displaystyle 
\frac{\partial}{\partial l} \ F \left[ \ b \ \right] = F \left[ \ b \ \right] \cdot \nabla b \bigl( \ \psi \left( l , p \right) \ \bigr) \\
\end{array}
\end{equation}

\noindent With respective reference conditions, $F \left[ \ B \ \right] = I$ at $\lambda = 0$, and $F \left[ \ b \ \right] = I$ at $l = 0$.

The full generating vector field and its unit generating vector field are related by $B = \vert B \vert \ b$. Therefore, the covariant derivative may be written, respectively,
\begin{equation} \label{E:GradbDecomp}
\begin{array}{c}
\begin{split}
\displaystyle 
\nabla B &\bigl( \ \phi \left( \sigma , p \right) \ \bigr) \\
\displaystyle
&= \vert B \vert \bigl( \ \phi \left( \sigma , p \right) \ \bigr) \ \biggl( \ \nabla b \bigl( \ \phi \left( \sigma , p \right) \ \bigr) + \frac{b}{2} \ \frac{\nabla B^{2}}{B^{2}} \bigl( \ \phi \left( \sigma , p \right) \ \bigr) \ \biggr) \\
\\
\end{split}
\\
\begin{split}
\displaystyle 
\nabla b &\bigl( \ \psi \left( s , p \right) \ \bigr) \\
\displaystyle
&= \frac{1}{\vert B \vert \bigl( \ \psi \left( s , p \right) \ \bigr)} \ \biggl( \ \nabla B \bigl( \ \psi \left( s , p \right) \ \bigr) - \frac{B}{2} \ \frac{\nabla B^{2}}{B^{2}}  \bigl( \ \psi \left( s , p \right) \ \bigr) \ \biggr) \\
\end{split}
\end{array}
\end{equation}

\noindent Where act expansion is evaluated along their respective flow representations $\phi$ and $\psi$, and in both cases the magnitude of the vector field is given by $\vert B \vert = \left( \ B \cdot B \ \right)^{1/2}$.

Substituting relations (\ref{E:GradbDecomp}) into the differential equations (\ref{E:SystemRepresentation}),
\begin{equation} \label{E:FullGVFRepresentation}
\begin{array}{c}
\begin{split}
\displaystyle 
&\frac{\partial}{\partial \lambda} \ F \left[ \ B \ \right] \\
\displaystyle
&= \vert B \vert \bigl( \ \phi \left( \sigma , p \right) \ \bigr) \ F \left[ \ B \ \right] \cdot \biggl( \ \nabla b \bigl( \ \phi  \left( \sigma , p \right) \ \bigr) + \frac{b}{2} \ \frac{\nabla B^{2}}{B^{2}} \bigl( \ \phi \left( \sigma , p \right) \ \bigr) \ \biggr) \\
\\
\end{split}
\\
\begin{split}
\displaystyle 
&\frac{\partial}{\partial l} \ F \left[ \ b \ \right] \\
\displaystyle
&= \frac{1}{\vert B \vert \bigl( \ \psi \left( s , p \right) \ \bigr)} \ F \left[ \ b \ \right] \cdot \biggl( \ \nabla B \bigl( \ \psi \left( s , p \right) \ \bigr) - \frac{B}{2} \ \frac{\nabla B^{2}}{B^{2}}  \bigl( \ \psi \left( s , p \right) \ \bigr) \ \biggr) \\
\end{split}
\end{array}
\end{equation}

Hence, by Theorem (\ref{Thm:RefGradLagrangian}), the path-ordered exponential solutions to the respective equations (\ref{E:SystemRepresentation}) and (\ref{E:FullGVFRepresentation}) must be equivalent,
\begin{equation} \label{E:ComponentRepresentationsSolutions}
\begin{array}{c}
\begin{split}
\displaystyle
&F \left[ \ B \ \right] = \mathcal{P} \text{exp} \int^{\lambda}_{0} d\sigma \ \nabla B \bigl( \ \phi \left( \sigma , p \right) \ \bigr) \\
 \displaystyle
 &= \mathcal{P} \text{exp} \ \int^{\lambda}_{0} d\sigma \ \vert B \vert \bigl( \ \phi \left( \sigma , p \right) \ \bigr) \ \biggl( \ \nabla b \bigl( \ \phi  \left( \sigma , p \right) \ \bigr) + \frac{b}{2} \ \frac{\nabla B^{2}}{B^{2}} \bigl( \ \phi \left( \sigma , p \right) \ \bigr) \ \biggr) \\
\\
\end{split}
\\
\begin{split}
\displaystyle
&F \left[ \ b \ \right]= \mathcal{P} \text{exp} \int^{l}_{0} ds \ \nabla b \bigl( \ \psi \left( s , p \right) \ \bigr) \\
\displaystyle
&= \mathcal{P} \text{exp} \ \int^{l}_{0} \ \frac{ds}{\vert B \vert \bigl( \ \psi \left( s , p \right) \ \bigr)} \ \biggl( \ \nabla B \bigl( \ \psi \left( s , p \right) \ \bigr) - \frac{B}{2} \ \frac{\nabla B^{2}}{B^{2}} \bigl( \ \psi \left( s , p \right) \ \bigr) \ \biggr) \\
\\
\end{split}
\end{array}
\end{equation}

The affine transformation (\ref{E:ArcLengthCondition}) remunerating the connectivity parameter is a diffeomorphism between the full generating vector field representation and the the arc length representation, $\phi \left( \lambda , p \right) \rightarrow \psi \left( l , p \right)$. Hence, equations (\ref{E:ComponentRepresentationsSolutions}) become (\ref{E:RepresentationsRelatations}),
\begin{equation} \label{E:ComponentRepresentationsSolutions2}
\begin{array}{c}
\begin{split}
\displaystyle
F \left[ \ B \ \right] &= \mathcal{P} \text{exp} \int^{\lambda}_{0} d\sigma \ \nabla B \bigl( \ \phi \left( \sigma , p \right) \ \bigr) \\
 \displaystyle
 &= \mathcal{P} \text{exp} \ \int^{l}_{0} ds \ \biggl( \ \nabla b \bigl( \ \psi \left( s , p \right) \ \bigr) + \frac{b}{2} \ \frac{\nabla B^{2}}{B^{2}} \bigl( \ \psi \left( s , p \right) \ \bigr) \ \biggr) \\
\\
\end{split}
\\
\begin{split}
\displaystyle
F \left[ \ b \ \right] &= \mathcal{P} \text{exp} \int^{l}_{0} ds \ \nabla b \bigl( \ \psi \left( s , p \right) \ \bigr) \\
\displaystyle
&= \mathcal{P} \text{exp} \ \int^{\lambda}_{0} d\sigma \ \biggl( \ \nabla B \bigl( \ \phi \left( \sigma , p \right) \ \bigr) - \frac{B}{2} \ \frac{\nabla B^{2}}{B^{2}} \bigl( \ \phi \left( \sigma , p \right) \ \bigr) \ \biggr) \\
\\
\end{split}
\end{array}
\end{equation}

We proceed to prove the remaining portion of the Lemma by noting, in general, the (free) Lie algebras generated by the integral terms on the respective LHS of equations (\ref{E:RepresentationsRelatations}) are not abelian. However, in the special case that the commutators,
\begin{equation} \label{E:Commutators}
\begin{array}{c}
\displaystyle
\biggl[ \ \int^{l}_{0} ds \ \nabla b \bigl( \ \psi \left( s , p \right) \ \bigr) \ , \ \int^{l}_{0} ds \ \frac{b}{2} \ \frac{\nabla B^{2}}{B^{2}} \bigl( \ \psi \left( s , p \right) \ \bigr) \ \biggr]  = 0 \\
\\
\displaystyle
\biggl[ \ \int^{\lambda}_{0} d\sigma \ \nabla B \bigl( \ \phi \left( \sigma , p \right) \ \bigr) \ , \ -\int^{\lambda}_{0} d\sigma \ \frac{B}{2} \ \frac{\nabla B^{2}}{B^{2}} \bigl( \ \phi \left( \sigma , p \right) \ \bigr) \ \biggr] = 0 \\
\\
\end{array}
\end{equation}

\noindent The relations (\ref{E:RepresentationsRelatationDecomposition}) follow immediately from an application of the BCH Theorem (see Appendix \ref{A:BCHTheorem}) to equations (\ref{E:ComponentRepresentationsSolutions}),
\begin{equation} \label{E:RepresentationDecomposition}
\begin{array}{c}
\begin{split}
\displaystyle
&F \left[ \ B \ \right] \\
\displaystyle
&= \mathcal{P} \text{exp} \biggl( \ \int^{l}_{0} ds \ \nabla b \bigl( \ \psi \left( s , p \right) \ \bigr) \ \biggr) \ \mathcal{P} \text{exp} \biggl( \ \int^{l}_{0} ds \ \frac{b}{2} \ \frac{\nabla B^{2}}{B^{2}} \bigl( \ \psi \left( s , p \right) \ \bigr) \ \biggr) \\
\\
\end{split}
\\
\begin{split}
\displaystyle
&F \left[ \ b \ \right] \\
\displaystyle
&= \mathcal{P} \text{exp} \biggl( \ \int^{\lambda}_{0} d\sigma \ \nabla B \bigl( \ \phi \left( \sigma , p \right) \ \bigr) \ \biggr) \ \mathcal{P} \text{exp} \biggl( \ -\int^{\lambda}_{0} d\sigma \ \frac{B}{2} \ \frac{\nabla B^{2}}{B^{2}} \bigl( \ \phi \left( \sigma , p \right) \ \bigr) \ \biggr) \\
\\
\end{split}
\end{array}
\end{equation}

\end{proof}

Given a general affine transformation remunerating the connectivity parameter, lemma (\ref{lem:RepresentationRelations}) amounts to a change of integration variable prescription; see equations (\ref{E:ComponentRepresentationsSolutions}) and equations (\ref{E:ComponentRepresentationsSolutions2}). As such, this lemma applies to any field of scalar multiples of the generating vector field and associated flow representation.

\subsection{Eulerian Specification of the Referential Gradient}
\label{SS:RefGradDynamics}

In this section, we derive an equivalent Eulerian specification of the referential gradient; that is the sixteen, coupled, non-linear partial differential equations which govern the referential gradient field dynamics in a coordinate chart. Furthermore, there is no identified closed form general solution in the Eulerian specification akin to equation (\ref{A:PexpReferentialGradient}).

We establish the coordinate representation of the rate of change with respect to the connectivity parameter, via the covariant derivative along the direction of the generating vector field given in coordinates. The following proposition follows as a consequence of equation (\ref{E:MCFlowCoords}), 
\begin{prop}
The rate of change with respect to connectivity parameter is given by the intrinsic directional derivative operator in the direction of the generating vector field,
\begin{equation} \label{E:ROCConnParam}
\displaystyle
\frac{\partial}{\partial \lambda} \mapsto \frac{D}{D\lambda} \equiv \nabla_{B} \\
\end{equation}
\end{prop}

\begin{proof}
Applying operator (\ref{E:ROCConnParam}) to $\phi \left( \lambda , p \right)$ yields the identity,
\begin{equation} \label{E:ROCConnParamProof}
\displaystyle
\frac{D}{D\lambda} \ \phi \left( \lambda , p \right) = \nabla_{B} \ \phi \left( \lambda , p \right) = i_{B} \circ \nabla \ \phi \left( \lambda , p \right) = B \bigl( \ \phi \left( \lambda , p \right) \ \bigr) \\
\end{equation}

\end{proof}

Recall the remark of section (\ref{SS:GeneralConnectivityMap}) following equation (\ref{E:MCFlowCoords}), in a coordinate chart $\mathbb{U}_{I}$ the units of the flow $\phi \left( \lambda , p \right)$ are identified with positions $x_{I} \bigl( \ \phi \left( \lambda , p \right) \ \bigr) \in \mathbb{U}_{I}$. Hence, in the chart $\mathbb{U}_{I}$,
\begin{equation} \label{E:ROCConnParamProofCoords}
\begin{array}{c}
\displaystyle
\frac{D}{D\lambda} \ x^{\mu}_{I} \bigl( \ \phi \left( \lambda , p \right) \ \bigr) = B^{\mu}_{I} \Bigl( \ x_{I} \bigl( \ \phi \left( \lambda , p \right) \ \bigr) \ \Bigr) \\
\end{array}
\end{equation}

\begin{ex}
In a four-dimensional spacetime coordinate chart $\mathbb{U}_{I}$, identify the generating field with the four-velocity field $U^{\mu}_{I} \left( \ x_{I} \ \right)$; the associated connectivity parameter is then identified with proper time $\tau$. Hence, equation (\ref{E:ROCConnParam}) is the intrinsic derivative operator with respect to proper time (see e.g., Ref. \cite{MihalasWeibel1984} $\S$ 4.2)
\begin{equation} \label{Ex:4VelIntrinsicDeriv}
\displaystyle
\frac{D}{D\tau} = U^{\nu}_{I} \left( \ x_{I} \ \right) \ \nabla_{\nu} \\
\end{equation}

\noindent For simplicity, assume Minkowski spacetime and Cartesian coordinates; $g = \text{diag} \left( -1 , 1 , 1 , 1 \right)$ and $\Gamma \left( \ x_{I} \ \right) = 0$. Thus, the four-velocity is given by $U^{\mu}_{I} \left( \ x_{I} \ \right) = \gamma \ \bigl( \ c \ , \ \mathbf{v} \ \bigr)$. Then, in the non-relativistic limit $\gamma \rightarrow 1$, equation (\ref{Ex:4VelIntrinsicDeriv}) reduces to the standard material derivative,
\begin{equation} \label{Ex:4VelIntrinsicDeriv1}
\displaystyle
\frac{D}{D\tau} = \frac{\partial}{\partial t} + \mathbf{v} \cdot \nabla \\
\end{equation}
\end{ex}

\begin{ex}
Identify the generating field with the four-magnetic field given by equation (\ref{Ex:FourMagneticField}). For simplicity, assume Minkowski spacetime and Cartesian coordinates; $g = \text{diag} \left( -1 , 1 , 1 , 1 \right)$ and $\Gamma \left( \ x_{I} \ \right) = 0$. Hence,
\begin{equation} \label{Ex:4MagFieldMink}
\displaystyle
B^{\mu}_{I} \left( \ x_{I} \ \right) = \gamma \ \biggl( \ \frac{\mathbf{B} \cdot \mathbf{v}}{c} \ , \ \mathbf{B}^{i} - \frac{\left( \ \mathbf{v} \times \mathbf{E} \ \right)^{i}}{c^{2}} \ \biggr) \\
\end{equation}

\noindent Where $\mathbf{E} = \mathbf{E} \left( \ t , \mathbf{x} \ \right)$ and $\mathbf{B} = \mathbf{B}  \left( \ t , \mathbf{x} \ \right)$ respectively denote the three-space electric and magnetic vector fields and $\mathbf{v}$ is the three-velocity of the observer, and $\gamma$ is the Lorentz factor.

Hence, the intrinsic derivative operator with respect to the connectivity parameter is,
\begin{equation} \label{Ex:4MagFieldIntrinsicDeriv}
\displaystyle
\frac{D}{D\lambda} = \gamma \ \frac{\mathbf{B} \cdot \mathbf{v}}{c^{2}} \ \frac{\partial}{\partial t} + \gamma \ \biggl( \ \mathbf{B}^{j} - \frac{\left( \ \mathbf{v} \times \mathbf{E} \ \right)^{j}}{c^{2}} \ \biggr) \ \nabla_{j} \\
\end{equation}

\noindent Which, in the non-relativistic limit $\vert \ \mathbf{v} \ \vert << c$, reduces to the (stationary) directional derivative along the field,
\begin{equation} \label{Ex:4MagFieldNonRelLimitIntrinsicDeriv}
\displaystyle
\frac{D}{D\lambda} = \mathbf{B} \cdot \nabla \\
\end{equation}
\end{ex}

Having identified the coordinate representation of the rate of change with respect to connectivity parameter, the following theorem is immediately obvious.

\begin{thm}
\emph{(Eulerian Specification of the Referential Gradient)}
\label{Thm:RefGradDynamics}
Let $B : \mathbb{M} \rightarrow T \mathbb{M}$ be a smooth, complete generating vector field with associated flow $\phi : \mathbb{R} \times \mathbb{M} \rightarrow \mathbb{M}$. The Eulerian specification of the referential gradient $F : I_{r} \times \mathcal{B}_{r} \rightarrow GL \left( 4 , \mathbb{R} \right)$, for $\mathcal{B}_{r} \subseteq \mathbb{M}$ and $I_{r} \subseteq \mathbb{R}$, is the coordinate representation of the Lie derivative with respect to the generating vector field,
\begin{equation} \label{Thm:DirRefGradCommutator}
\displaystyle
\mathcal{L}_{B \left( \ \phi \left( \lambda , p \right) \ \right)} \ F \left( \ \lambda , p \ \right) \cdot \hat{h} = \biggl[ \ B \bigl( \ \phi \left( \lambda , p \right) \ \bigr) \ , \ F \left( \ \lambda , p \ \right) \cdot \hat{h} \ \biggr] \\
\end{equation}

\noindent for arbitrary reference shift direction vector $\hat{h} \in T_{p} \mathbb{M}$, and all $\lambda \in I_{r}$.

Moreover, if the metric connection is torsion-free, the Eulerian specification of the referential gradient reduces to,
\begin{equation} \label{Thm:DirRefGradCommutator2}
\displaystyle
\mathcal{L}_{B \left( \ \phi \left( \lambda , p \right) \ \right)} \ F \left( \ \lambda , p \ \right) \cdot \hat{h} = 0 \\
\end{equation}
\end{thm}

\begin{proof}
The Eulerian specification of the referential gradient follows upon application of proposition (\ref{E:ROCConnParam}) to equation (\ref{E:DirectionalRefGradHolonomyProof}).
\begin{equation} \label{E:DirRefGradCommutatorProof}
\displaystyle
B \Bigl( \ \phi \bigl( \lambda , p \bigr) \ \Bigr) \cdot \nabla \ \Bigl( \ F \left( \ \lambda , p \ \right) \cdot \hat{h} \ \Bigr) = \Bigl( \ F \left( \ \lambda , p \ \right) \cdot \hat{h} \ \Bigr) \cdot \nabla \ B \Bigl( \ \phi \bigl( \lambda , p \bigr) \ \Bigr) \\
\end{equation}

To construct the coordinate representation of equation (\ref{E:DirRefGradCommutatorProof}), let $\mathcal{B}_{r} \left( p \right) \subset \mathbb{U}_{I}$. By equation (\ref{E:RefGradHolonomy}) of Theorem (\ref{Thm:RefGradLagrangian}), in the coordinate chart $\mathbb{U}_{I}$ the lineal element $\Bigl( \ \lambda , x_{I} \left( p \right) \ \Bigr) \in I_{r} \times \mathbb{U}_{I}$ characterizing the Lagrangian specification of the referential gradient at a distance $\lambda \in I_{r}$ from the reference point $x_{I} \left( p \right) \in \mathbb{U}_{I}$, is equivalent to the point $x_{I} \Bigl( \ \phi \left( \lambda , p \right) \ \Bigr) \in \mathbb{U}_{I}$; hence,
\begin{equation} \label{E:LinealElementEquiv}
\displaystyle
F^{\mu}_{\nu}  \Bigl( \ \lambda , x_{I} \left( p \right) \ \Bigr) = F^{\mu}_{\nu} \Bigl( \ x_{I} \bigl( \ \phi \left( \lambda , p \right) \ \bigr) \ \Bigr) \\
\end{equation}

\begin{rem}
For ease of notation, in what follows we represent the general position $x_{I} = x_{I} \bigl( \ \phi \left( \lambda , p \right) \ \bigr)$ at variable $\lambda \in I_{r}$ in the coordinate chart $\mathbb{U}_{I}$. Hence, the position at $\lambda = 0$ may be written as $x_{0 I} = x_{I} \bigl( \ \phi \left( 0 , p \right) \ \bigr)$. 
\end{rem}

In the coordinate chart $\mathbb{U}_{I}$, therefore, the $\hat{e}^{I}_{\mu}$ component of equation (\ref{E:DirRefGradCommutatorProof}) is
\begin{equation} \label{E:DirRefGradCoordDynamics}
\begin{array}{c}
\displaystyle 
B^{\alpha}_{I} \left( \ x_{I} \ \right) \ \nabla_{\alpha} \ \Bigl( \ F^{\mu}_{\nu} \left( \ x_{I} \ \right) \ \hat{h}^{\nu}_{I} \ \Bigr) = \Bigl( \ F^{\alpha}_{\nu} \left( \ x_{I} \ \right) \ \hat{h}^{\nu}_{I} \ \Bigr) \ \nabla_{\alpha} \ B^{\mu}_{I} \left( \ x_{I} \ \right) \\
\end{array}
\end{equation}

\noindent Expanding the covariant derivatives in equation (\ref{E:DirRefGradCoordDynamics}),
\begin{equation} \label{E:CovDerDirRefGradCoord}
\begin{split}
\displaystyle
B^{\alpha}_{I} \left( \ x_{I} \ \right) \ \frac{\partial}{\partial x^{\alpha}_{I}} \ \Bigl( \ F^{\mu}_{\nu} &\left( \ x_{I} \ \right) \ \hat{h}^{\nu}_{I} \ \Bigr) - \Bigl( \ F^{\alpha}_{\nu} \left( \ x_{I} \ \right) \ \hat{h}^{\nu}_{I} \ \Bigr) \ \frac{\partial}{\partial x^{\alpha}_{I}} \ B^{\mu}_{I} \left( \ x_{I} \ \right) \\
\displaystyle
& = T^{\mu}_{\alpha \rho} \left( \ x_{I} \ \right) \ \Bigl( \ F^{\alpha}_{\nu} \left( \ x_{I} \ \right) \ \hat{h}^{\nu}_{I} \ \Bigr) \ B^{\rho}_{I} \left( \ x_{I} \ \right) \\
\end{split}
\end{equation}

\noindent where the torsion $T^{\mu}_{\rho \alpha} \left( \ x_{I} \ \right) = 2 \ \Gamma^{\mu}_{[ \rho \alpha ]} \left( \ x_{I} \ \right)$. The LHS of equation (\ref{E:CovDerDirRefGradCoord}) is simply the Lie bracket of the vector fields $B \left( \ x_{I} \ \right)$ and $F \left( \ x_{I} \ \right) \cdot \hat{h}$. 

Hence, the $\mu^{th}$-component of the Lie derivative of the vector $F \left( \ x_{I} \ \right) \cdot \hat{h}$ with respect to the generating vector field $B \left( \ x_{I} \ \right)$ is given by,
\begin{equation} \label{E:LieDerivCoords}
\displaystyle
\Bigl( \ \mathcal{L}_{B_{I} \left( x_{I} \right)} \ F \left( \ x_{I} \ \right) \cdot \hat{h}_{I} \ \Bigr)^{\mu} = T^{\mu}_{\alpha \rho} \left( \ x_{I} \ \right) \ \Bigl( \ F^{\alpha}_{\nu} \left( \ x_{I} \ \right) \ \hat{h}^{\nu}_{I} \ \Bigr) \ B^{\rho}_{I} \left( \ x_{I} \ \right)
\end{equation}

\noindent 
Moreover, equation (\ref{Thm:DirRefGradCommutator2}) follows immediately for a torsion-free metric connection $T^{\mu}_{\rho \alpha} = 0$.

\end{proof}

In the practical application of Theorem (\ref{Thm:RefGradDynamics}), it is useful to use the Levi-Civita connection. Furthermore, recall the generating vector field $B \left( \ x_{I} \ \right)$ is assumed known \textit{a-priori} everywhere. Noting the reference shift direction vector $\hat{h}_{I} = \hat{h}^{\mu}_{I} \ \hat{e}^{I}_{\mu} \in T_{p} \mathbb{U}_{I}$ is an arbitrary albeit fixed (i.e., coordinate-independent) vector, the Eulerian specification of the referential gradient is given by the following set of (sixteen) coupled first-order, partial differential equations,
\begin{equation} \label{E:RefGradDynamicsLC}
\displaystyle
B^{\alpha}_{I} \left( \ x_{I} \ \right) \ \frac{\partial F^{\mu}_{\nu} \left( \ x_{I} \ \right)}{\partial x^{\alpha}_{I}} - F^{\alpha}_{\nu} \left( \ x_{I} \ \right) \ \frac{\partial B^{\mu} \left( \ x_{I} \ \right)}{\partial x^{\alpha}_{I}} = 0 \\
\end{equation}

\noindent with reference condition $F^{\mu}_{\nu} \left( \ x_{0 I} \ \right) = \delta^{\mu}_{\nu}$. 

\begin{ex}
The Eulerian specification of the referential gradient in Minkowski space, $g_{\mu \nu} \left( \ x_{I} \ \right) = \text{diag} \left( -1 , 1 , 1 , 1 \right)$, satisfies the following set of coupled first-order partial differential equations,
\begin{equation} \label{E:RefGradDynamicsExpand}
\begin{array}{l}
\begin{split}
\displaystyle
\frac{B^{0}_{I} \left( \ x_{I} \ \right)}{c} &\ \frac{\partial}{\partial t} \ F^{0}_{\nu} \left( \ x_{I} \ \right) + B^{j}_{I} \left( \ x_{I} \ \right) \ \nabla_{j} \ F^{0}_{\nu} \left( \ x_{I} \ \right) \\
\displaystyle
&- \frac{1}{c} \ \frac{\partial B^{0}_{I} \left( \ x_{I} \ \right)}{\partial t} \ F^{0}_{\nu} \left( \ x_{I} \ \right) = F^{j}_{\nu} \left( \ x_{I} \ \right) \ \nabla_{j} \ B^{0}_{I} \left( \ x_{I} \ \right) \\
\\
\end{split}
\\
\begin{split}
\displaystyle
\frac{B^{0}_{I} \left( \ x_{I} \ \right)}{c} \ \frac{\partial}{\partial t} \ F^{i}_{\nu} \left( \ x_{I} \ \right) + \Bigl[ \ B_{I} \left( \ x_{I} \ \right) \ , & \ F_{\nu} \left( \ x_{I} \ \right) \ \Bigr]^{i} \\
\displaystyle
&= \frac{F^{0}_{\nu} \left( \ x_{I} \ \right)}{c} \ \frac{\partial B^{i}_{I} \left( \ x_{I} \ \right)}{\partial t} \\
\end{split}
\end{array}
\end{equation}

\noindent With reference conditions at the point $x_{0I} = x^{\mu}_{0I} \left( p \right)$ given by, 
\begin{equation} \label{E:RefGradDyanBndryCond}
\begin{array}{c c c c}
F^{0}_{0} \left( \ x_{0 I} \ \right) = 1 & F^{0}_{i} \left( \ x_{0 I} \ \right) = 0 & F^{i}_{0} \left( \ x_{0 I} \ \right) = 0 & F^{i}_{j} \left( \ x_{0 I} \ \right) = \delta^{i}_{j}
\end{array}
\end{equation}
%
\end{ex}

\section{Summary}
\label{S:Summary}

This paper introduces the \textit{Referential Gradient of the Flow} of a vector field, a generally covariant measure of the geometric structure of the flow of a vector field in four-dimensional spacetime. We assume \textit{a-priori} the generating vector field exists, is everywhere smooth, and satisfies some set of governing evolution equations. The mathematical formalism of flows is provided as background from which the referential gradient object is defined.

We provided the explicit relation between the referential gradient of the flow and the generating vector field from two equivalent perspectives: a Lagrangian specification with respect to a generalized connectivity parameter, and an Eulerian specification making explicit the evolution dynamics at each point of the manifold. The Lagrangian specification Theorem (\ref{Thm:RefGradLagrangian}) yields a general closed-form functional solution with respect the generating vector field in terms of a generalized connectivity parameter. While the Eulerian specification Theorem (\ref{Thm:RefGradDynamics}) makes explicit the referential gradient dynamics at each point of the manifold. 

Due to the integration, we prove three transformation lemmas that identify the conditions under which the referential gradient transforms as a 1-1 tensor, forms a group with respect to the connectivity parameter, and the proper change of variable relations between the corresponding referential gradient representations. Lemma (\ref{lem:RefGrad}) proves manifest covariance of the referential gradient provided the closure of its domain of definition is contained within coordinate chart overlap. Lemma (\ref{lem:RefGradGroup}) identifies the necessary and sufficient conditions that the referential gradient forms a group with respect to translations of connectivity parameter; that is, in general, 1) the group property is a direct consequence of the similar group property of the flow, or 2) the representation of $\nabla B$ in any coordinate chart is independent of the connectivity parameter and reference point. Finally, since the flow of a vector field represents an equivalence class under affine transformations of the connectivity parameter, Lemma (\ref{lem:RepresentationRelations}) provides the proper relations associated with a change of integration variable; in particular, between the two most natural representations of the flow.

In a follow-up paper, we develop a geometric mechanics and thermodynamics using the Lagrangian specification of the referential gradient. In this context we explore the importance of Lemma (\ref{lem:RefGradGroup}) with respect to topological invariants such as the linking number, as well as the consequences of relaxing the smoothness assumption of the \textit{a-priori} generating vector field which lead to some powerful topological constraints on the storage, transport, and release of field energy in a system. Furthermore, in this context, Lemma (\ref{lem:RepresentationRelations}) will prove important in application to systems in which only partial information of the generating vector field may known.


\section{Acknowledgements} 
The author would like to acknowledge L. A. Fisk at the University of Michigan, and B. J. Lynch at the University of California at Berkeley for insightful discussions in the development of this research. This research was supported, in part, by NSF Grant AGS-1043012, and NASA LWS Grant NNX10AQ61G.

\appendix
\begin{appendix}
\section{Generalized First Order Taylor Theorem for Smooth Vector Fields}
\label{A:GeneralizedFirstOrderTaylorExpansion}

By standard Taylor expansion theorem (see e.g., Ref. \cite{Lee2006}, Theorem A.58), any smooth vector field $\xi \left( y \right)$ for any fixed $y \in \mathbb{M}$, may be written,
\begin{equation} \label{AE:GeneralTaylorExpansion}
\displaystyle
\xi \left( y + w \right) = \xi \left( y \right) + \nabla_{w} \ \xi \left( y \right) + \int^{1}_{0} ds \ \Bigl( \ \nabla_{w} \ \xi \left( y + sw \right) - \nabla_{w} \ \xi \left( y \right) \ \Bigr) \\
\end{equation}

\noindent Where $\nabla_{w} \ \xi \left( y \right)$ is the covariant derivative along the vector $w$. In a chart $\mathbb{U}_{I}$, the vector field $\xi \left( y \right) = \xi^{\mu}_{I} \bigl( \ x_{I} \left( y \right) \ \bigr) \ \hat{e}^{I}_{\mu} \left( y \right)$, and the covariant derivative along the vector $w = w^{\nu}_{I} \ \hat{e}^{I}_{\nu} \left( y \right)$ is given by,
\begin{equation} \label{AE:CovDeriv}
\displaystyle
\nabla_{w} \ \xi \left( y \right) = w^{\nu}_{I} \ \biggl( \ \frac{\partial \ \xi^{\mu}_{I} \bigl( \ x_{I} \left( y \right) \ \bigr)}{\partial x^{\nu}_{I} \left( y \right)} + \Gamma^{\mu}_{\nu \rho} \left( y \right) \ \xi^{\rho}_{I} \bigl( \ x_{I} \left( y \right) \ \bigr) \ \biggr) \ \hat{e}^{I}_{\mu} \left( y \right) \\
\end{equation}

\noindent Hence the $\hat{e}^{I}_{\mu} \left( y \right)$ component of equation (\ref{AE:GeneralTaylorExpansion}) may be written,
\begin{equation} \label{AE:TaylorExpansionWRemainderConst}
\begin{array}{l}
\begin{split}
\displaystyle
\xi^{\mu}_{I} \bigl( \ x_{I} \left( y + w \right) \ \bigr) &- \xi^{\mu}_{I} \bigl( \ x_{I} \left( y \right) \ \bigr) \\
\displaystyle
&= w^{\nu}_{I} \ \biggl( \ \frac{\partial \ \xi^{\mu}_{I} \bigl( \ x_{I} \left( y \right) \ \bigr)}{\partial x^{\nu}_{I} \left( y \right)} + \Gamma^{\mu}_{\nu \rho} \left( y \right) \ \xi^{\rho}_{I} \bigl( \ x_{I} \left( y \right) \ \bigr) \ \biggr) \\
\\
\end{split}
\\
\begin{split}
\displaystyle
+ \ w^{\nu}_{I} \ \int^{1}_{0} \biggl[ \ \biggl(  &\ \frac{\partial \ \xi^{\mu}_{I} \bigl( \ x_{I} \left( y + sw \right) \ \bigr)}{\partial x^{\nu}_{I} \left( y \right)} - \frac{\partial \ \xi^{\mu}_{I} \bigl( \ x_{I} \left( y \right) \ \bigr)}{\partial x^{\nu}_{I} \left( y \right)} \ \biggr) \\
\displaystyle
& + \Gamma^{\mu}_{\nu \rho} \left( y \right) \ \biggl( \ \xi^{\rho}_{I} \bigl( \ x_{I} \left( y + sw \right) \ \bigr) - \xi^{\rho}_{I} \bigl( \ x_{I} \left( y \right) \ \bigr) \ \biggr) \ \biggr] ds \\
\end{split}
\end{array}
\end{equation}
\noindent It follows immediately the integral remainder is equal to zero when the vector $w = 0$.

\section{Baker-Cambpell-Hausdorff Theorem}
\label{A:BCHTheorem}

The exponential of a matrix satisfies the identity (Dinkin's formula, see Ref. \cite{BonfuglioliFulci2012}),
\begin{equation} \label{E:DynkinIdentity}
\displaystyle
\text{exp} \left( X \right) \cdot \text{exp} \left( Y \right) = \text{exp} \biggl( \ X + Y + 1/2 \ \left[ \ X , Y \ \right] + \sum_{n} \ a_{n} \ C_{n} \bigl( \ X , Y \ \bigr) \ \biggr) \\
\end{equation}

\noindent Where the $a_{n}$ are constant coefficients, and the $C_{n} \bigl( \ X , Y \ \bigr)$ are homogeneous (Lie) polynomials in $X$ and $Y$ of degree $n$ (i.e., nested commutators). The first few $a_{n} \ C_{n} \bigl( \ X , Y \ \bigr)$ terms are well-known, and given by (see Ref. \cite{BonfuglioliFulci2012}),
\begin{equation} \label{E:CnXY}
\begin{array}{l l}
\displaystyle
a_{1} = 1/12 \ \ \ & \ \ \ C_{1} = \bigl[ \ X , \left[ \ X , Y \ \right] \ \bigr] + \bigl[ \ \left[ \ X , Y \ \right] , Y \ \bigr] \\
\\
\displaystyle
a_{2} = 1/24 \ \ \ & \ \ \ C_{2} = \bigl[ \ \bigl[ \ X , \left[ \ X , Y \ \right] \ \bigr] , Y \ \bigr] \\
\\
\displaystyle
a_{3} = 1/720 \ \ \ & \ \ \ \begin{split} C_{3} = \bigl[ \ X , \bigl[ \ X , \bigl[ &\ X , \left[ \ X , Y \ \right] \ \bigr] \ \bigr] \ \bigr] \\ \displaystyle &- \bigl[ \ \bigl[ \ \bigl[ \ \left[ \ X , Y \ \right] , Y \ \bigr] , Y \ \bigr] , Y \ \bigr] \\ \end{split}
\\
\end{array}
\end{equation}

\end{appendix}

\end{document}